\documentclass[11pt]{article}

\usepackage{amsmath,amsthm,nicefrac}
\usepackage{amsfonts, amstext}
\usepackage{mathtools}
 \usepackage{libertine}
 \usepackage[libertine]{newtxmath}

\usepackage{typearea}
\typearea{14}
\usepackage[font={small,it}]{caption}

\usepackage{setspace}

\usepackage[shortlabels]{enumitem}
\usepackage[breaklinks]{hyperref}
\hypersetup{colorlinks=true,%
            citebordercolor={.6 .6 .6},linkbordercolor={.6 .6 .6},%
citecolor=blue,urlcolor=black,linkcolor=blue,pagecolor=black}

\usepackage[nameinlink]{cleveref}
\Crefname{algocf}{Algorithm}{Algorithms}
\crefname{algocfline}{line}{lines}
\Crefname{invariant}{Invariant}{Invariants}
\Crefname{claim}{Claim}{Claims}
\Crefname{subclaim}{Subclaim}{Subclaims}

\usepackage{epsfig}
\usepackage{amsthm}

\usepackage{mathrsfs}
\usepackage{xspace}
\usepackage{soul}
\usepackage{latexsym}
\usepackage{bbm}

\usepackage{framed}

\usepackage[dvipsnames]{xcolor}
\definecolor{DarkGray}{rgb}{0.66, 0.66, 0.66}
\definecolor{DarkPowderBlue}{rgb}{0.0, 0.2, 0.6}
\definecolor{fluorescentyellow}{rgb}{0.8, 1.0, 0.0}
\definecolor{cerulean}{rgb}{0.0, 0.48, 0.65}
\definecolor{bleudefrance}{rgb}{0.19, 0.55, 0.91}

\usepackage[ruled,vlined,linesnumbered,algonl]{algorithm2e}
\SetEndCharOfAlgoLine{}
\SetKwComment{Comment}{\footnotesize$\triangleright$\ }{}

\SetCommentSty{mycommfont}

\usepackage{thmtools,thm-restate}

\allowdisplaybreaks

\newcounter{note}[section]

\sethlcolor{fluorescentyellow}

\newcommand{\initOneLiners}{%
    \setlength{\itemsep}{0pt}
    \setlength{\parsep }{0pt}
    \setlength{\topsep }{0pt}
}

\newtheorem{theorem}{Theorem}[section]
\newtheorem{lemma}[theorem]{Lemma}
\newtheorem{claim}[theorem]{Claim}
\newtheorem{fact}[theorem]{Fact}
\newtheorem{corollary}[theorem]{Corollary}

\theoremstyle{definition}
\newtheorem{defn}[theorem]{Definition}

\theoremstyle{remark}

\makeatletter 
\renewcommand{\theinvariant}{(I\@arabic\c@invariant)}
\makeatother

\newcommand{\tO}{\tilde{O}}

\newcommand{\var}{{\mathbb{V}}}
\newcommand{\ex}{{\mathbb{E}}}
\newcommand{\E}{\mathbb{E}}
\newcommand{\relv}{\eta}
\newcommand{\ugp}{u_G(p)}
\newcommand{\zgp}{z_G(p)}
\newcommand{\xgp}{x_G(p)}

\newcommand{\eps}{\varepsilon}

\newcommand{\EE}{\mathbb{E}}

\newcommand{\poly}{\operatorname{poly}}

\newcommand{\nf}{\nicefrac}

\newcommand{\junk}[1]{}
\newcommand{\eat}[1]{}

\newif\ifhideproofs

\ifhideproofs
\usepackage{environ}
\NewEnviron{hide}{}

\fi

\usepackage{fullpage}
\usepackage[utf8]{inputenc}

\title{Network Unreliability in Almost-Linear Time}

\author{Ruoxu Cen\\Duke University \and Jason Li\\Carnegie Mellon University \and Debmalya Panigrahi\\Duke University}
\date{}

\begin{document}

\maketitle

\begin{abstract}
    The network unreliability problem asks for the probability that a given undirected graph gets disconnected when every edge independently fails with a given probability $p$. Valiant (1979) showed that this problem is \#P-hard; therefore, the best we can hope for are approximation algorithms. In a classic result, Karger (1995) obtained the first FPTAS for this problem by leveraging the fact that when a graph disconnects, it almost always does so at a near-minimum cut, and there are only a small (polynomial) number of near-minimum cuts. Since then, a series of results have obtained progressively faster algorithms  to the current bound of $m^{1+o(1)} + \tilde{O}(n^{3/2})$ (Cen, He, Li, and Panigrahi, 2024). 
%
    In this paper, we obtain an $m^{1+o(1)}$-time algorithm for the network unreliability problem. This is essentially optimal, since we need $O(m)$ time to read the input graph. Our main new ingredient is relating network unreliability to an {\em ideal} tree packing of spanning trees (Thorup, 2001). 
\end{abstract}


\section{Introduction}
\label{sec:intro}
The unreliability of an undirected graph $G = (V, E)$, denoted $\ugp$, is the probability that it gets disconnected when every edge independently fails with a given probability $p$. It measures the robustness of a network to random edge failures (as against worst case failures, measured by the minimum cut). Reliability problems are extensively studied (see the books \cite{colbourn1987combinatorics, chaturvedi2016network}) and the problem of estimating $\ugp$ has been dubbed ``the most fundamental'' problem in this space~\cite{Karger20}. This problem was shown to be \#P-hard by Valiant in 1979~\cite{Valiant79}. The first fully polynomial-time randomized approximation scheme (FPRAS) was given in a seminal work by Karger \cite{Karger99}. For any constant $\eps \in (0, 1)$, this algorithm outputs a $(1\pm\eps)$-approximation for $\ugp$ with high probability (whp)\footnote{In this paper, as in the network unreliability literature, a result is said to hold with high probability if it holds with probability $1-\frac{1}{\poly(n)}$.} in $\tO(mn^4)$ time, where $n$ is the number of vertices in the graph and $m$ is the number of edges. The running time was later improved by Harris and Srinivasan \cite{HarrisS18} to $n^{3+o(1)}$ using a finer-grained analysis of Karger's algorithmic framework. Further improvements came in a series of works by Karger \cite{Karger16, Karger17, Karger20} that utilized progressively optimized versions of a recursive random contraction framework to eventually attain a running time of $\tO(n^2)$. This quadratic bound is a natural threshold because the techniques developed for network unreliability could also (implicitly or explicitly) enumerate the possibly $\Omega(n^2)$ minimum cuts of the graph. The quadratic barrier was recently breached by Cen, He, Li, and Panigrahi~\cite{CenHLP24} who obtained a running time of $m^{1+o(1)} + \tO(n^{3/2})$ by using an importance sampling technique. The ultimate goal in this line of work, explicitly conjectured by Karger in \cite{Karger20}, is to obtain a network unreliability algorithm that runs in (essentially) linear time in the size of the input graph. This last bound is also the best possible, since any network unreliability algorithm must read its entire input.

\subsection{Our Result}

In this paper, we give a fully polynomial-time randomized approximation scheme (FPRAS) for the network unreliability problem that has a running time of $m^{1+o(1)}$. We show the following theorem:

\begin{theorem}\label{thm:main}
    For any \eat{constant} $\eps\in (0, 1)$, there is a randomized Monte Carlo algorithm for the network unreliability problem that runs in $m^{1+o(1)}\eps^{-O(1)}$ time and outputs a $(1\pm\eps)$-approximation to $\ugp$ whp.
\end{theorem}

Up to lower order terms, the above theorem brings the line of research into faster network unrelibility algorithms to a close. One might still hope to improve on the lower order terms: specifically, improve $m^{o(1)}$ to $\poly\log(m)$ in the running time. As we will see, the $m^{o(1)}$ overhead in our running time comes from two sources: The first is the use of maximum flow algorithms, the fastest of which run in $m^{1+o(1)}$ time~\cite{ChenKLPGS22,BrandCPKLGSS23}; this will automatically improve if faster maximum flow algorithms are discovered. The second source is more intrinsic to our method and relates to the amount of variance in the estimator used in our work; it is an interesting question as to whether this variance can be reduced from $m^{o(1)}$ to $\poly\log(m)$ in our framework.


\subsection{Our Techniques}

Intuitively, we divide into two cases depending on the value of $\ugp$: if (roughly) $\ugp < \nf{1}{n}$, we say the graph is reliable, else it is unreliable. Morally, we would like to use Monte Carlo sampling for the unreliable case, i.e., sample edge failure outcomes and use an indicator variable for disconnection to estimate $\ugp$. But, we cannot do this directly since generating each independent sample takes $\tO(m)$ time, and we need $\tO(n)$ samples. Instead, we use the recursive contraction algorithm of Karger and Stein~\cite{karger1996new} to (intuitively) generate $\tO(n)$ samples in a correlated manner such that each sample only takes $\tO(1)$ amortized time. We explain how the recursive contraction algorithm can be used to generate Monte Carlo samples in the next paragraph.

A mechanism for generating a Monte Carlo sample is to randomly order edges, contract a prefix of edges in this order, and ask if the resulting graph is a singleton vertex. If this process is independently repeated multiple times, it produces a set of independent samples. But, now, suppose that instead of creating completely independent random permutations, we create correlations in the following manner: we first select a random set of edges that appear at the head of the permutations for all the samples; then, we create two independent copies and independently sample the next set of edges in these two copies; then, we subdivide further by creating two independent copies of each of the previous copies, and so on. This creates a correlated set of permutations on which we perform contractions to generate samples of the network. The advantage is that this latter process can be implemented very efficiently, because this is exactly the recursive contraction algorithm~\cite{karger1996new}. Normally, a recursive contraction algorithm would generate $\tO(n^2)$ samples in $\tO(n^2)$ time. However, using the assumption that $\ugp > \nf{1}{n}$, we show that $\tO(n)$ samples generated by this process are sufficient, and that this can be produced in $m^{1+o(1)}$ time.

Next, we consider the reliable case, i.e., when $\ugp < \nf{1}{n}$. Recall that we would like to estimate the probability that at least one cut fails (i.e., all edges are deleted). In the reliable case, this is a rare event and can be well-approximated by the expected number of failed cuts, i.e., the corresponding union bound. Furthermore, the large cuts can be ignored (their failure is unlikely), but we are still left with $\Omega(n^2)$ cuts, e.g., the minimum cuts themselves. To do this in close to $O(m)$ time, we sample a surrogate distribution $q(C)$ that approximates the real distribution $p^{|C|}$. The distribution $q(C)$ is defined via a distribution on spanning trees, from which we sample a spanning tree, and perform a sampling procedure on this tree to define a cut. We describe the spanning tree distribution next. 

We use a fractional spanning tree packing that was introduced by
Thorup~\cite{Thorup01, Thorup08}. 
The total value of the packing equals the min-ratio value of all multiway cuts in the graph, i.e., $\pi := \min_{\text{multiway cuts } C} \frac{|C|}{k(C)-1}$, where $k(C)$ is the number of sides of $C$. The packing is also ``edge-disjoint'' in the sense that each edge is used only in spanning trees that add up to a total value at most $1$. This latter property ensures that every edge in the min-ratio cut appears in exactly $\nf{1}{\pi}$ fraction of the trees by value. More generally, the tree packing also ensures that this property holds recursively on the induced subgraph defined on each side of the min-ratio cut, which has a larger cut ratio than $\pi$.
When everything is settled, any edge $e$ will be in $\nf{1}{\pi(e)}$ fraction of trees for some cut ratio $\pi(e)$. We call $\pi(e)$ the {\em level} of edge $e$, and the entire collection of fractional trees is called an {\em ideal} tree packing.

So, the first step of our sampling is to generate a spanning tree from the ideal tree packing distribution. Next, we need to sample edges in the sampled spanning tree to define the sampled cut. For intuition, let us imagine an idealized scenario where for any cut $C$, every spanning tree $T$ in the ideal tree packing contains $\frac{|C_i|}{\pi_i}$ edges at level $i$, 
where $C_i$ is the set of level-$i$ edges in $C$. Then, an unbiased estimator for $p^{|C|}$ can be obtained by sampling every edge $e$ in a sampled spanning tree $T$ with probability $p^{\pi(e)}$. The advantage of this process is that we are sampling a small number of edges that appear in the spanning tree, much smaller than the size of the cut, which can be done much more efficiently. This intuition is not precise because edges of a cut are not proportionately distributed across the spanning trees of the ideal tree packing. Nevertheless, we show that the error due to the non-uniformity (i.e., variance) can be bounded.

A final complication is that we need to use an approximation algorithm for ideal tree packing since exact algorithms are too slow for our purpose. This means the levels of edges are now only approximately correct, which leads to dependence between the random sampling process across the different (approximate) levels. We show that this additional error can be controlled for the above arguments to still hold approximately, but in the process, we lose additional sub-polynomial terms in the running time to counter the variance of the estimator.

\subsection{Overall Algorithm and Paper Organization} 
\label{sec:organization}


We now give an overview of the algorithm. 
Our overall algorithm builds a recursive computation tree. There are several base cases that we describe first. The first two base cases also appear in prior work~\cite{Karger20,CenHLP24}. The interesting (and new) base case is the third one, which is also the main contribution of this paper. 
\begin{itemize}
    \item If $n \le \eps^{-O(1)}$, then any $n^{O(1)}\eps^{-O(1)}$ running time is still $\eps^{-O(1)}$, so we run Karger's algorithm \cite{Karger20} that gives an unbiased estimator for $\ugp$. The following theorem states the properties of this estimator
    (relative variance of a random variable $X$ is $\frac{\var[X]}{(\E[X])^2}$): 
    \begin{theorem}[\cite{Karger20}]\label{thm:karger-estimator}
        Given a graph $G$ with vertex size $n$, an unbiased estimator of $\ugp$ with $O(1)$ relative variance can be computed in $\tO(n^2)$ time. As a consequence, a $(1\pm\eps)$-approximation to $\ugp$ can be computed in time $\tO(n^2\eps^{-2})$.
    \end{theorem} 

    \item The second base case is to run Monte Carlo sampling when $p \ge \theta$ for some threshold $\theta$ whose value is given by \Cref{lem:z-approx-u}. The properties of the estimator are summarized below:
            \begin{lemma}[Lemma 1.3 in \cite{CenHLP24}]\label{lem:mc}
                For any $p\ge\theta$, an unbiased estimator of $\ugp$ with relative variance $O(1)$ can be computed in time $m^{1+o(1)}$. As a consequence, a $(1\pm\eps)$-approximation to $\ugp$ can be computed in $m^{1+o(1)}\eps^{-2}$ time under the condition that $p\ge\theta$.
            \end{lemma}
    We are left to define the value of $\theta$. We set $\theta$ such that when $p < \theta$, it is very unlikely that more than one cut fails. To quantify this, define $\zgp$ as the expected number of failed cuts $\zgp := \sum_{C_i} p^{|C_i|}$ and $\xgp$ as the expected number of failed cut pairs $\xgp := \sum_{C_i \ne C_j} p^{|C_i\cup C_j|}$. By the inclusion-exclusion principle, $\ugp \ge \zgp - \xgp$. We define $\theta$ as follows: 
    \begin{lemma}[Phase transition, Lemma 2.1 of \cite{CenHLP24}]\label{lem:z-approx-u}
    There exists a threshold $\theta$ such that $u_G(\theta) = n^{-\Omega(1/\log \log n)}$ satisfying the following property: When $p<\theta$, we have $\frac{x_G(p)}{z_G(p)} \le \frac{1}{\log n}$, and therefore,\footnote{All logarithms are base $e$ unless otherwise stated.}
    \[
    \left(1-\frac{1}{\log n}\right) z_G(p) \le \ugp \le z_G(p).
    \]
    \end{lemma}
    Intuitively, this means that when $p < \theta$, we can use an estimator of $\zgp$ as a surrogate for that of $\ugp$, although we need to be careful about the $\nf {1}{\log n}$ gap between $\zgp$ and $\ugp$.
    
    \item {\em The final base case is the most interesting new contribution of this paper.} It is invoked when $\ugp < n^{-1-\Omega(1/\log \log n)}$. In this case, we run our new algorithm (\Cref{sec:packing}) and
    prove the following lemma (an estimator $X$ for $\ugp$ with relative bias $\delta$ satisfies $\E[X] \in (1\pm \delta)\ugp$):
    \begin{restatable}{lemma}{importance}\label{lem:interface-importance}
        If $\ugp < n^{-1-\Omega(1/\log\log n)}$, an estimator for $\ugp$ with relative bias $n^{-\Omega(1)}$ and relative variance $\le 1$ can be computed in $m^{1+o(1)}$ time. As a consequence, a $(1\pm\eps)$-approximation to $\ugp$ can be computed in $m^{1+o(1)} \eps^{-O(1)}$ time under the condition that $\ugp < n^{-1-\Omega(1/\log\log n)}$.
    \end{restatable}
    
\end{itemize}

We have described the base cases, all of which are non-recursive algorithms. The remaining case is when $\ugp \ge n^{-1-O(1/\log\log n)}$ and $p < \theta$. In this case, we run a step of recursive contraction (\Cref{sec:contraction}). As in \cite{CenHLP24}, we interleave recursive contraction calls to a standard sparsification algorithm. However, unlike prior work, we obtain a sharp tradeoff between the rate of progress and the amount of variance in this recursion, by leveraging the parameters given by the ideal tree packing. We obtain the following lemma (the relative second moment of a random variable $X$ is defined as $\EE[X^2]/(\EE[X])^2$):
\begin{restatable}{lemma}{contract}\label{lem:interface-contract}
    Suppose $\ugp \ge n^{-1-O(1/\log\log n)}$ and $p < \theta$. An estimator $X$ for $\ugp$ with relative bias $n^{-\Omega(1)}$ and relative second moment $(\log n)^{O(\log \log n)}$ can be computed in $m^{1+o(1)}$ time.
\end{restatable}

We now show that \Cref{thm:main} follows from these lemmas:
\begin{proof}[Proof of \Cref{thm:main}]
The base cases are immediate from above. 
So, we focus on the recursive contraction case.
By setting a sufficiently large constant in the exponent of the first base case, we can ensure in the recursive case that the $n^{-\Omega(1)}$ upper bound of relative bias is less than $\frac{\eps}{2}$.
Let $X$ be the estimator output by the recursive contraction algorithm.
From \Cref{lem:interface-contract}, we have that the relative variance $\eta[X]\le n^{o(1)}$. 
By standard techniques (see \Cref{lem:mc-sample}), we can run the algorithm $n^{o(1)}\eps^{-2}$ times to get a $(1\pm\frac{\eps}{2})$-approximation of $\E[X]$ whp.
Because $\E[X]$ is a $(1\pm n^{-\Omega(1)})$-approximation of $\ugp$, the aggregated estimator is a $(1\pm \eps)$-approximation of $\ugp$.
Each run takes $m^{1+o(1)}\eps^{-O(1)}$ time, so the overall running time is also $m^{1+o(1)}\eps^{-O(1)}$.
\end{proof}

Finally, we describe how the algorithm decides which case it is in at any node of the computation tree. The first base case can be identified based on the size of the graph. To identify the second base case, we use an indirect method as in \cite{CenHLP24} since we do not know the value of $\theta$. We first run the Monte Carlo algorithm and calculate the estimator of $\ugp$ given by this algorithm (which is the empirical probability of disconnection). If the value of this estimator is at least $n^{-o(1)}$, then we can conclude that the estimator $(1\pm\eps)$-approximates $\ugp$ whp. If the estimator returns a smaller value, then we are in the case $p < \theta$. Next, we need to determine if we are in the third base case. This is more complicated since we do not know the value of $\ugp$, and cannot estimate it by Monte Carlo simulation since we are in a reliable case.
Instead of detecting the condition on $\ugp$ directly, we (approximately) detect the consequence of this condition using two surrogate conditions. The first surrogate condition given in \Cref{def:surrogate-criteria} is based on a lower bound of $\zgp$ (\Cref{lem:ugp-ktau}). 
It either decides we are not in the third case, or it allows us to obtain a tight estimator $\tilde{z}$ of $\zgp$ with relative bias $O(\nf 1n)$. 
Since we have $p<\theta$, $\tilde{z}$ is already an estimator of $\ugp$ with relative bias $O(\nf {1}{\log n})$, by \Cref{lem:z-approx-u}. Still, we need a second condition to achieve a better relative bias of $n^{-\Omega(1)}$.
The second surrogate condition we use is $\tilde{z} < \frac{1}{2n}$, which is a close approximation of the original condition $\ugp < n^{-1-\Omega(1/\log\log n)}$. Combined together, the two conditions either decide that we are not in the third base case, or imply that the conclusions of \Cref{lem:interface-contract} hold.


\medskip\noindent{\bf Roadmap.}
We set up basic notation and discuss some preliminary facts in \Cref{sec:prelim}.
In \Cref{sec:packing}, we give the importance sampling algorithm using an ideal tree packing to handle the reliable case (the third base case) --  this is the main contribution of this paper.
In \Cref{sec:contraction}, we give the recursive contraction algorithm for the unreliable case, show the sharp tradeoff between progress and variance in the recursion, and combine it with all the base cases.

\section{Preliminaries}
\label{sec:prelim}
In the network unreliability problem, we are given an undirected multigraph $G$, and an edge failure probability $p\in(0,1)$. The goal is to $(1\pm \eps)$-approximate the unreliability $u_G(p)$, which is the probability that $G$ disconnects when each edge is independently deleted with probability $p$. 

Throughout the paper, we use $\lambda$ to denote the value of a minimum cut of the graph. Note that $\ugp = \Pr[\cup_{\text{cuts } C}\, C \text{ disconnects}]$, and the probability that a cut $C$ disconnects is given by $p^{|C|}$.

\noindent {\bf Relative Variance and Bias.} We give estimators for $\ugp$ (and $\zgp$) in this paper. Here,  we give (standard) definitions of relative variance and bias of an estimator:
\begin{defn}
The {\em relative variance} of a random variable $X$, denoted $\relv[X]$, is defined as the ratio of its variance and its squared expectation, i.e., $\relv[X] = \frac{\var[X]}{\ex^2[X]} = \frac{\ex[X^2]}{\ex^2[X]}-1$. 
We also define {\em relative second moment} of $X$ as $\frac{\ex[X^2]}{\ex^2[X]} = \relv[X]+1$.
\end{defn}

Similar to variance, relative variance can be decreased by taking multiple independent samples.

\begin{fact}[Lemma I.4 of \cite{Karger17}]\label{fact:rel-var-decrease}
The relative variance of the average of $N$ independent samples of $X$ is $\frac{\eta[X]}{N}$.
\end{fact}
This leads to the following property:
\begin{lemma}[Lemma I.2 of \cite{Karger17}]\label{lem:mc-sample}
Fix any $\eps, \delta \in (0, 1)$. For a random variable $X$ with relative variance $\relv[X]$,  the median of $O\left(\log\frac{1}{\delta}\right)$ averages of $O\left(\frac{\relv[X]}{\eps^2}\right)$ independent samples of $X$ is a $(1\pm \eps)$-approximate estimate of $\E[X]$ with probability $1-\delta$.
\end{lemma}

Next, we define bias and relative bias of an estimator:
\begin{defn}
    Let $X$ be an estimator of $u$. Then, the bias of $X$ is $|\E[X]-u|$, and assuming $u > 0$, its relative bias is $\frac{|\E[X]-u|}{u}$.
\end{defn}

\smallskip\noindent{\bf A Coarse Bound on $\ugp$.} We will use the following known bound:
    \begin{fact}[Corollary II.2 of \cite{Karger16}]\label{fact:ugp-basic-range}
        For any graph on $n$ vertices with mincut value $\lambda$ and any $p\in (0, 1)$, we have $p^\lambda \le \ugp \le n^2 p^\lambda$.
    \end{fact}

\smallskip\noindent{\bf $k$-Strong Components.} 
In the reliable case, it is unlikely that a well-connected part of the graph disconnects; therefore, these parts can be contracted before estimating $\ugp$. To formalize this, recall that a $k$-strong component is a maximal $k$-edge-connected induced subgraph. The strength $k_e$ of an edge $e$ is the maximum value $k$ such that a $k$-strong component contains $e$.
\eat{
Recall the following properties~\cite{BenczurK15}:
\begin{enumerate}
    \item Suppose the strength of $e_1$ is larger than $e_2$. Then, deleting $e_2$ does not change the strength of $e_1$, and contracting $e_1$ does not change the strength of $e_2$.
    \item The $k$-strong components of a graph, for all $k$, form a laminar family.
    \item $\sum_{e\in E(G)} w_e/k_e \le n_G-1$.
\end{enumerate}
}
The next lemma shows that $3\lambda$-strong components can be contracted before estimating $\ugp$:
\begin{lemma}\label{lem:contract-strong-oneshot}
    Suppose $p^\lambda \le \eps n^{-1}$.
    Then, we can contract all $3\lambda$-strong components to form $H$, such that $u_H(p)$ $(1\pm\eps)$-approximates $u_G(p)$.
\end{lemma}
\begin{proof}
    Let the $3\lambda$-strong components be $S_1, S_2, \ldots, S_r$. Define $H$ by contracting the sets $S_1, S_2, \ldots, S_r$ in $G$. Let $D[G]$ be the event that $G$ disconnects when each edge fails independently with probability $p$. 
    Clearly, $D[G]$ implies $D[H]\lor \bigvee_i D[G[S_i]]$, where $G[S_i]$ is the induced subgraph on $S_i$. 
    By the union bound,
    \[u_G(p) \le u_H(p) + \sum_i u_{G[S_i]}(p).\]
    We also have $D[H]$ implies $D[G]$, so $u_G(p) \ge u_H(p)$. It remains to prove that $\sum_i u_{G[S_i]}(p) \le \eps \cdot u_G(p)$.

    Let $n_i, \lambda_i$ respectively denote the number of vertices and mincut value of $G[S_i]$. We have $$\frac{\sum_i u_{G[S_i]}(p)}{u_G(p)} \le \frac{\sum_i n_i^2 p^{\lambda_i}}{p^\lambda} \le  p^{3\lambda-\lambda} \sum_i n_i^2 \le n^2 p^{2\lambda} \le \eps^2 < \eps,$$
    where the first inequality uses \Cref{fact:ugp-basic-range}.
\end{proof}

\section{The Reliable Case: Importance Sampling on Ideal Tree Packing}
\label{sec:packing}

This section solves the reliable case (the third base case from \Cref{sec:organization}), whose main lemma is restated below.
\importance*

We will first design an importance sampling algorithm for a slightly different goal, which is an estimator for $\zgp$ instead of $\ugp$.
\begin{lemma}\label{lem:interface-importance-zgp}
    If $\ugp < n^{-1-\Omega(1/\log\log n)}$, an estimator for $\zgp$ with relative bias $O(1/n)$ and relative variance $\le 1$ can be computed in $m^{1+o(1)}$ time.
\end{lemma}
We will use \Cref{lem:interface-importance-zgp} to prove \Cref{lem:interface-importance} in \Cref{sec: bias-z-to-u}. The rest of this section is devoted to \Cref{lem:interface-importance-zgp}.

One complexity here is that the algorithm does not know the value of $\ugp$, so it cannot determine if the assumption is true. Instead, we will use an alternative assumption based on tree packings that implies the original assumption $\ugp\le n^{-1-\Omega(1/\log \log n)}$.

At a high level, our algorithm uses importance sampling. Note that we need to preferentially sample the near-minimum cuts since they are more likely to fail. Moreover, we need to sample a cut efficiently, e.g., by only picking $O(1)$ edges from the cut. To attain these two properties, we use a fractional spanning tree packing to guide our sampling process. The advantage is that it distributes a near-minimum cut among $\Omega(\lambda)$ spanning trees, creating an $O(1)$-size projection on average in each of them. If we can sample the edges in the projection from a spanning tree, we can recover the near-minimum cut without having to explicitly sample all edges of the cut. 

One limitation of the above approach is there is a duality gap of 2 between the tree packing value and min-cut value $\lambda$. 
As a consequence, to sample a near min-cut, we need to pick $\ge 2$ tree edges on average, which results in a sampling space of $\Omega(n^2)$. 
Since we can only afford $o(n^2)$ samples, a naive approach yields only $o(1)$-fraction of these near min-cuts; it is not hard to create examples where the resulting estimator has large error because it misses the critical near min-cuts that determine the value of $\ugp$. 

To alleviate this concern, consider contracting well-connected vertex sets in the graph. This preserves near min-cuts; so sampling is valid on the contracted graph. Importantly, every vertex pair that remains has a near min-cut separating it; this gives us a different structural tool for evaluating how the sample hits the near min-cuts. More generally, if we do these contractions at multiple recursive levels that preserve near min-cuts of different values, than the resulting graphs allow us a tighter control over the set of cuts that we hit in our sampling process. Since we need the tree packing to guide cut sampling across multiple recursive levels, we need it to evenly partition the cuts for different levels with their corresponding ranges of connectivity. This is achieved by an {\em ideal tree packing} that we describe next.

\subsection{Ideal Tree Packing}
Let a \emph{fractional tree packing} be a distribution over spanning trees of a graph. The \emph{load} of each edge is the probability that a tree sampled from this distribution contains that edge. A maximum tree packing is a distribution that minimizes the maximum load of any edge. The \emph{min-ratio cut} of a graph is the partition $P$ of the vertices minimizing the ratio $\frac{d(P)}{|P|-1}$, where $d(P)$ is the cut value $P$. In other words, $d(P)=|\partial P|$, where $\partial P$ is the set of edges whose endpoints are separated by $P$.
A well-known theorem of Nash-Williams~\cite{nash1961edge} and Tutte~\cite{Tutte61} states that the maximum load of any edge in a maximum tree packing equals the inverse of the minimum value of the cut ratio defined above.

Thorup's ideal tree packing~\cite{Thorup01,Thorup08} is a fractional tree packing that is not only (globally) a maximum tree packing, but also maximum in a local sense. The loads $\ell^*(e)$ of edges $e$ in an ideal tree packing are defined by the following \emph{recursive min-ratio cut process}. 
\begin{defn}[Recursive min-ratio cut process]
    Let $P$ be the min-ratio cut of the input graph, i.e., the vertex partition minimizing $\frac{d(P)}{|P|-1}$.
    Define $\ell^*(e) := \frac{|P|-1}{d(P)}$ for every edge $e\in\partial P$.
    Recursively apply this process on the induced subgraph of each non-singleton component of $P$ to set the loads on the other edges.
\end{defn}

We call $\ell^*(e)$ the ideal load of edge $e$. The following lemma ensures that this process uniquely produces the loads of a fractional tree packing.
\begin{lemma}[Lemma 3 of \cite{Thorup08}]\label{lem:ideal-loads-def}
The following hold:
\begin{enumerate}    
    \item The ideal loads are unique (although the min-ratio cuts chosen by the recursive min-ratio cut process may not be unique).
    \item $\ell^*$ is in the spanning tree polytope, i.e., it is the load of a fractional tree packing.
\end{enumerate}
\end{lemma}

Let $\pi^*$ be the minimum ratio $\frac{d(P)}{|P|-1}$ among all vertex partitions $P$.
The first iteration of the recursive min-ratio cut process described above defines the ideal loads of the edges in the min-ratio cut to be $\frac{1}{\pi^*}$.
We have the following relation between $\pi^*$ and $\lambda$.
\begin{fact}[Lemma 3 of \cite{Thorup01}]\label{fact:pi-range}
    $\frac{\lambda}{2} < \pi^* \le \lambda$.
\end{fact}

A key property of the recursive min-ratio cut process is that the ideal loads are monotone non-increasing in the recursion.
\begin{lemma}[Lemma 14 of \cite{Thorup01}]\label{lem:ideal-loads-monotone}
    Suppose $P$ is a min-ratio cut during the recursive min ratio cut process, $e_1\in \partial P$, and $e_2$ is in an induced subgraph of a component of $P$. Then, $\ell^*(e_1) \ge \ell^*(e_2)$.
\end{lemma}
It follows that the ideal loads defined in the first iteration is the maximum throughout the recursion, i.e., $\frac{1}{\pi^*} = \max_e \ell^*(e)$.

Let $k_\tau$ be the number of vertices if we contract all edges with $\ell^*(e)<\frac{1}{\tau}$. Clearly, $k_\tau$ is monotonically increasing in $\tau$. We give a lower bound on $\zgp$ in terms of $k_\tau$:
\begin{lemma}\label{lem:ugp-ktau}
$\zgp \ge k_\tau p^{2\tau}$ for any $\tau \ge \pi^*$. 
\end{lemma}
\begin{proof}
    Form a graph $H$ by contracting all edges $e$ with $\ell^*(e) < \frac{1}{\tau}$.  Since $\ell^*(e)=\frac{1}{\pi^*} \ge \frac{1}{\tau}$ for edges $e$ in the global min-ratio cut, $H$ is not a singleton.
    By definition of $k_\tau$, we have $|V(H)|=k_\tau$.

Next, we claim that if we run the recursive min-ratio cut process in $H$, the ideal loads of edges in $H$ are the same as their original ideal loads in $G$. Denote the ideal loads on $H$ by $\ell_H$.
Whenever the recursive min-ratio cut process on $G$ finds a partition $P$ with cut ratio $\frac{d(P)}{|P|-1} \le \tau$, all edges in $\partial P$ have load $\ell^*(e) \ge \frac{1}{\tau}$ and will not be contracted. So, the cut $P$ is preserved in $H$ with the same cut ratio, and it must be a min-ratio cut in the current induced subgraph. By \Cref{lem:ideal-loads-def} (1), we can choose the same cut in the process to define $\ell_H$.
Whenever the recursive min-ratio cut process on $G$ finds  a partition $P$ with cut ratio $\frac{d(P)}{|P|-1} > \tau$, the cut edges $e\in \partial P$ have $\ell^*(e) < \frac{1}{\tau}$. Moreover, by \Cref{lem:ideal-loads-monotone}, all edges in the current induced subgraph have $\ell^*(e) < \frac{1}{\tau}$. So, the whole induced subgraph will be contracted in $H$.
By inductively applying the above two arguments, we can define a recursive min-ratio cut process on $H$ that produces $\ell_H(e) = \ell^*(e)$ for any $e\in E(H)$. By \Cref{lem:ideal-loads-def} (2), $\ell_H$ is in the spanning tree polytope. So,  $\sum_{e\in E(H)}\ell^*(e)=|V(H)|-1=k_\tau-1$.

Let  $d(v)$ be the degree of $v$, and $\bar d$ be the average degree in $H$. We bound
    \[\bar{d} = \frac{1}{k_\tau}\sum_{v\in V(H)} d(v) = \frac{2|E(H)|}{k_\tau} \le \frac{2}{k_\tau} \sum_{e\in E(H)}\ell^*(e) \cdot \tau = 2\tau \cdot \frac{k_\tau-1}{k_\tau}< 2\tau.\]
Since the degree cut of each vertex $v\in V(H)$ corresponds to a separate cut in $G$, we have
    \[\zgp \ge \sum_{v\in V(H)} p^{d(v)} \ge k_\tau p^{\bar{d}} \ge k_\tau p^{2\tau},\]
where the middle inequality follows by convexity of the function $f(x)=p^x$.
\end{proof}

\medskip\noindent{\bf Approximate Tree Packing.} 
The recursive algorithm used to define ideal loads above requires many min-ratio cut computations; we don't know how to implement them all in almost-linear time.
Instead, we approximate the ideal loads using a \emph{greedy tree packing} algorithm described below. The main purpose is to eventually define surrogate parameters $\tilde{k}_\tau$ that can substitute for the parameters $k_\tau$ in \Cref{lem:ugp-ktau}.

\begin{defn}[Greedy tree packing]\label{def:greedy-packing}
Given a input graph $G$ and a number of rounds $k$, the greedy algorithm iteratively finds a collection of spanning trees $\{T_1, T_2, \ldots, T_k\}$, such that each $T_i$ is a minimum spanning tree on $G$ with respect to edge loads induced by $\{T_1, \ldots, T_{i-1}\}$, which we define below.
For a collection of spanning trees $\mathcal T$, define its induced loads $\ell^{\mathcal{T}}$ by the edge loads of the uniform distribution over trees in $\mathcal T$, i.e., $\ell^{\mathcal T}(e) = \frac{|\{T\in\mathcal{T}:T\ni e\}|}{|\mathcal T|}$. (Specially, for an empty $\mathcal T$, $\ell^T(e) = 0$.)

\end{defn}
The next lemma asserts that this greedy process approximates an ideal tree packing:
\begin{lemma}\label{cor:greedy-packing}
In an unweighted graph with min-cut value $\lambda$, given a parameter $\delta\in(0, 1)$, we can find a (multi-)set $\mathcal T$ of $6\lambda\ln m/\delta^2$ spanning trees such that
 \begin{enumerate}
 \item $|\ell^{\mathcal{T}}(e) - \ell^*(e)| \le 3\delta/\lambda, ~\forall e\in E$, where $\ell^{\mathcal T}$ are the edge loads for the uniform distribution over $\mathcal T$.
 \item Each tree $T$ is the MST of some weight function $\ell_T$ satisfying $|\ell_T(e)-\ell^*(e)|\le\delta/\lambda$ for all $e\in E$.
 \end{enumerate}
The algorithm runs in time $\tilde{O}(m/\delta^2)$. The trees are implicit, but the loads $\ell^{\mathcal T}$ are output explicitly.
\end{lemma}
\begin{proof}
We use the following lemma from \cite{Thorup01}:
\begin{lemma}[Proposition 16 of \cite{Thorup01}]\label{lem:greedy-packing}
In an unweighted graph with min-cut value $\lambda$, suppose that $\mathcal{T}$ is a (multi-)set of spanning trees generated by greedy tree packing after $\ge 6\lambda \ln m/\delta^2$ iterations, where parameter $\delta < 2$. Let $\ell^{\mathcal T}$ be the edge loads for the uniform distribution over $\mathcal T$. Then, for all edges $e$,
    $$|\ell^{\mathcal{T}}(e) - \ell^*(e)| \le \delta/\lambda.$$
\end{lemma}

Now, we proceed to prove \Cref{cor:greedy-packing}. 
Compute a tree packing of $N=12\lambda\ln m/\delta^2$ trees, which is twice as many as needed by \Cref{lem:greedy-packing}, and let $\mathcal T$ be the latter $N/2$ of the spanning trees computed. On each iteration $i>6\lambda\ln m/\delta^2$ of greedy tree packing, we have $|\ell^{\mathcal T_i}(e)-\ell^*(e)|\le\delta/\lambda$ for all $e\in E$, where $\mathcal T_i$ is the current tree packing of $i-1$ trees. For the next tree $T$ to be packed, define $\ell_T=\ell^{\mathcal T_i}$, which satisfies property~(2). For property~(1), for each $e\in E$ we have
\[ \ell^{\mathcal T}(e)=\frac{N\cdot\ell^{\mathcal T_N}(e)-N/2\cdot\ell^{\mathcal T_{N/2}}(e)}{N/2}=2\ell^{\mathcal T_N}(e)-\ell^{\mathcal T_{N/2}}(e) ,\]
which together with $|\ell^{\mathcal T_N}(e)-\ell^*(e)|\le\delta/\lambda$ and $|\ell^{\mathcal T_{N/2}}(e)-\ell^*(e)|\le\delta/\lambda$ concludes property~(1).

For the running time, we dynamically maintain the minimum spanning tree over the iterations using, e.g.\ the algorithm of~\cite{holm01}. Each iteration takes $\tilde O(1)$ time to compute the new minimum spanning tree, and then $\tilde O(n)$ time to update the weights on each edge of the spanning tree. The total running time over the $\tilde O(\lambda/\delta^2)$ iterations is $\tilde O(\lambda/\delta^2)\cdot\tilde O(n)=\tilde O(m/\delta^2)$ since $m\ge n\lambda/2$.
\end{proof}

We wish to obtain a multiplicative approximation of $\ell^*(e)$ instead of the additive approximation in \Cref{cor:greedy-packing}. Moreover, we will later group edges multiplicatively by loads and need to ensure that there are few groups. For both these purposes, we contract all $3\lambda$-strong components in the input graph. The next lemma says that this restricts $\ell^*(e)$ to $(\frac1{3\lambda},\frac2\lambda)$ for all the uncontracted edges:
\begin{lemma}\label{lem:load-range}
After contracting all $3\lambda$-strong components, each edge $e$ satisfies $\ell^*(e)\in (\frac1{3\lambda},\frac2\lambda)$.
\end{lemma}
\begin{proof}
By \Cref{fact:pi-range}, $1/\pi^* < 2/\lambda$. Since edge loads are monotonically decreasing in the recursive ideal load algorithm, each edge has ideal load less than $2/\lambda$.

For a given uncontracted edge $e$ of strength less than $3\lambda$, consider an instance $G[S]$ in the recursive load algorithm whose min-ratio cut sets value $\ell^*(e)$. Since $G[S]$ contains edge $e$ of strength $< 3\lambda$,  the min-cut value of $G[S]$ is $< 3\lambda$. By definition of $\ell^*(e)$, the minimum cut ratio of $G[S]$ is $\frac{1}{\ell^*(e)}$. Any (2-way) cut is a partition whose cut ratio equals the cut value. So, the connectivity of $G[S]$ is no less than the value of the min-ratio cut. Denoting by $\lambda_{G[S]}$ the minimum cut of $G[S]$, we conclude that $\frac{1}{\ell^*(e)} \le \lambda_{G[S]} < 3\lambda$, which implies $\ell^*(e) > \frac{1}{3\lambda}$.
\end{proof}

Contracting all $3\lambda$-strong components is without loss of generality due to \Cref{lem:contract-strong-oneshot}.
(If the assumption of \Cref{lem:contract-strong-oneshot} does not hold, we can directly switch to the unreliable case since that assumption is implied by the assumption of reliable case.)
By applying \Cref{cor:greedy-packing} with parameter $\delta/3$ on the contracted graph given by \Cref{lem:load-range}, we get $\ell^{\mathcal T}(e)\approx_{\delta}\ell^*(e)$ for all $e\in E$, where $x \approx_\delta y$ denotes $(1+\delta)^{-1}y \le x \le (1+\delta)y$.

\medskip\noindent{\bf Edge Layering.}
We set the parameter $\delta := \Theta(1/\log \log n)$ and partition the edges into $O(\frac1\delta) = O(\log \log n)$ levels based on their ideal load as follows: Let $\pi_i = (1+\delta)^i\lambda/2$; we define $E_i\subseteq E$ to be the edges $e\in E$ with $\ell^*(e) \in [\frac{1}{\pi_i}, \frac{1}{\pi_{i-1}})$. Since $\ell^*(e) \in (\frac{1}{3\lambda}, \frac{2}{\lambda})$ for all edges $e$, the sets $E_1,E_2,\ldots,E_L$ partition $E$ for $L=O(\frac1\delta)$. To avoid clutter, define $E_{>i}=E_{i+1}\cup\cdots\cup E_L$. Recall that $k_{\pi_i}$ is the number of nodes after contracting all edges in $E_{>i}$ (i.e., edges with $\ell^*(e) < \frac{1}{\pi_i}$).

One complication is that the values $k_\tau$ are defined on the ideal tree packing, which means that we do not know them exactly. So, we next define approximations $\widetilde E_i$ and $\tilde k_\tau$ based on our approximate tree packing $\mathcal T$. Define $\widetilde E_i\subseteq E$ to be the edges $e\in E$ with $\ell^{\mathcal T}(e)\in[\frac{1}{\pi_i}, \frac{1}{\pi_{i-1}})$, and let $\tilde k_\tau$ be the number of nodes after contracting all edges with $\ell^{\mathcal T}(e)<\frac1\tau$. As earlier, define $\widetilde E_{>i}=\widetilde E_{i+1}\cup\cdots\cup\widetilde E_L$, so $\tilde k_{\pi_i}$ is the number of nodes after contracting all edges in $\widetilde E_{>i}$ (i.e., edges with $\ell^{\mathcal T}(e) < \frac{1}{\pi_i}$).

For a spanning tree $T\in\mathcal T$, define $\tilde k_{\pi_i}(T) = |E(T)\cap\widetilde E_i|$; 
note that since $T$ is a tree, $\tilde k_{\pi_i}(T)$ is also the number of nodes in $T$ after contracting the set of edges $E(T)\setminus \widetilde E_{i}$, minus 1. We now show how these parameters $k_{\pi_i}$, $\tilde k_{\pi_i}$, and $\tilde k_{\pi_i}(T)$ are related.

\begin{lemma}\label{lem:tilde-ki-range}
For every $i$, we have $k_{\pi_i}\le\tilde k_{\pi_{i+1}}\le k_{\pi_{i+2}}$ and $\tilde k_{\pi_i}(T)\le\tilde k_{\pi_{i+2}}-1$ for all $T\in\mathcal T$.
\end{lemma}
\begin{proof}
By property~(1) of \Cref{cor:greedy-packing}, we have $|\ell^{\mathcal{T}}(e) - \ell^*(e)| \le 3\eps$ for all edges $e$.
By property~(2) of \Cref{cor:greedy-packing}, each tree $T\in\mathcal T$ is the MST of some weight function $\ell_T$ satisfying $|\ell_T(e)-\ell^*(e)|\le\eps$ for all $e\in E$. Combining the two properties, we obtain $|\ell_T(e)-\ell^{\mathcal T}(e)|\le4\eps$ for all $e\in E$. We can set $\eps=\Theta(\delta/\lambda)$ small enough that $\ell^{\mathcal T}(e)\approx_\delta\ell^*(e)$ and $\ell_T(e)\approx_\delta\ell^{\mathcal T}(e)$ for all $e\in E$. 

For the first statement, observe that
\[ e\in\widetilde E_{>i+1}\implies\ell^{\mathcal T}(e)<\frac1{\pi_{i+1}}\implies\ell^*(e)<\frac1{\pi_i}\implies e\in E_{>i} ,\]
so $\widetilde E_{>i+1}\subseteq E_{>i}$, and it follows that $k_{\pi_i}\le\tilde k_{\pi_{i+1}}$. The other inequality $\tilde k_{\pi_{i+1}}\le k_{\pi_{i+2}}$ follows by an identical argument.

For the second statement, we have
\begin{align*}
e\in\widetilde E_{>i+2}&\implies \ell^{\mathcal T}(e) < \frac{1}{\pi_{i+2}} \implies \ell_T(e) < \frac{1}{\pi_{i+1}} , \quad\text{and}
\\ e\in\widetilde E_{i} &\implies \ell^{\mathcal T}(e) \ge \frac{1}{\pi_i} \implies \ell_T(e) \ge \frac{1}{\pi_{i+1}}.
\end{align*}
Imagine running Kruskal's MST algorithm on weight function $\ell_T$. All edges in $\widetilde E_{>i+2}$ are prioritized over all edges in $\widetilde E_i$, so by the time an edge in $\widetilde E_i$ is visited by Kruskal's algorithm, the current spanning forest has at most $\tilde k_{\pi_{i+2}}$ connected components. It follows that there are at most $k_{\pi_{i+2}}-1$ edges in $\widetilde E_i$ added to $T$.
\end{proof}

\subsection{Surrogate Criterion for Deciding Reliable vs Unreliable Cases}\label{sec:surrogate}
Based on the edge layering above, we define the criterion that is used to distinguish between the reliable and unreliable cases. 
This replaces the condition $\ugp < n^{-1-\Omega(1/\log\log n)}$ in \Cref{lem:interface-importance-zgp}. The latter condition, in conjunction with \Cref{lem:z-approx-u} and \Cref{lem:ugp-ktau}, yield a condition on $k_\tau$ that, in principle, can be used directly to separate the reliable and unreliable cases. However, since we do not have a sufficiently fast algorithm to compute the parameters $k_\tau$, we will instead use the approximate parameters $\tilde{k}_\tau$ and rely on \Cref{lem:tilde-ki-range} to relate the $\tilde{k}_\tau$ and $k_\tau$ parameters.

\begin{defn}\label{def:surrogate-criteria}
    The algorithm is in the reliable case if $n^{1+30\delta}\tilde k_{\pi_i}p^{2\pi_i}\le1$ holds for all $i$. Otherwise, it is in the unreliable case. (Note that since we are not in the second base case, we assume $p<\theta$ for both these cases.)
\end{defn}


The next two lemmas summarize the relevant implications of this condition being true or false, i.e., in the unreliable and reliable cases respectively. 

Recall that $\ell^*(\cdot)$ and $\ell^{\mathcal T}(\cdot)$ respectively denote the ideal loads and approximate loads computed by greedy tree packing (as in \Cref{cor:greedy-packing}) for all edges. We first extend this notation to set functions on cuts as follows: For any function $\ell(\cdot)$ (such as $\ell^*(\cdot), \ell^{\mathcal T}(\cdot)$, etc.)\ and a cut $C$, we use the shorthand $\ell(C)$ to denote $\sum_{e\in C} \ell(e)$.

Let us first consider the case when \Cref{def:surrogate-criteria} is false, i.e., the unreliable case. In this case, we show that $\ugp > n^{-1-O(1/\log\log n)}$ holds. This implies that in order to prove \Cref{lem:interface-importance-zgp}, it suffices to show the conclusion of the lemma under the surrogate condition \Cref{def:surrogate-criteria}.

\begin{lemma}[Unreliable Case]\label{lem:criteria-inverse}
    If $n^{1+30\delta}\tilde k_{\pi_i} p^{2 \pi_i} > 1$ for some $i\ge 0$ and $p<\theta$, then $\ugp > n^{-1-O(1/\log \log n)}$.
\end{lemma}
\begin{proof}
First, we consider the case when $\tilde{k}_{\pi_i} = 1$. The assumption becomes $p^{2\pi_i} > n^{-1-30\delta}$.
Since $\pi_i \ge \pi_0=\lambda/2$, this implies
$p^{\lambda}\ge p^{2\pi_i} > n^{-1-30\delta}$, which further implies $\ugp \ge n^{-1-O(1/\log \log n)}$.

Next, we consider the case when $\tilde{k}_{\pi_i} > 1$. 
Recall that $\tilde{k}_{\pi_i}$ is the number of nodes after contracting edges with $\ell^{\mathcal{T}}(e) < \frac{1}{\pi_i}$. So, in this case, there exists an edge $e$ with $\ell^{\mathcal{T}}(e) \ge \frac{1}{\pi_i}$. Since $\ell^{\mathcal{T}}(e)\approx_\delta \ell^*(e)$, we have $\pi^* \le \frac{1}{\ell^*(e)} \le (1+\delta)\pi_i = \pi_{i+1}$. 
So, we can apply \Cref{lem:ugp-ktau} to get $\zgp \ge k_{\pi_{i+1}} p^{2\pi_{i+1}}$.

    From the assumption, we have
    \[1 <  n^{1+30\delta} \tilde k_{\pi_i} p^{2 \pi_i} \le n^{1+30\delta} k_{\pi_{i+1}} p^{2 \pi_i}=n^{1+30\delta}k_{\pi_{i+1}}p^{2\pi_{i+1}/(1+\delta)},\]
    where the second inequality uses \Cref{lem:tilde-ki-range}.
    Combined with \Cref{lem:ugp-ktau}, we have   
    \begin{align*}
        (\zgp)^{1/(1+\delta)} \ge (k_{\pi_{i+1}})^{1/(1+\delta)}p^{2\pi_{i+1}/(1+\delta)}
        =(k_{\pi_{i+1}})^{-O(\delta)}k_{\pi_{i+1}}p^{2\pi_{i+1}/(1+\delta)}\ge n^{-O(\delta)}n^{-(1+30\delta)} \ge n^{-1-O(\delta)}.
    \end{align*}
    Therefore, $\zgp > n^{-1-O(\delta)}$ as claimed. Combined with \Cref{lem:z-approx-u} that $\ugp \ge (1-\frac{1}{\log n})\zgp$ when $p\le \theta$, we have $\ugp \ge n^{-1-O(1/\log \log n)}$.
\end{proof}

We are now left to show that under \Cref{def:surrogate-criteria}, we can prove the conclusion of \Cref{lem:interface-importance-zgp}; this will establish \Cref{lem:interface-importance-zgp}.
%

Define $\mathcal C$ as the collection of cuts $C$ with $\ell^{\mathcal T}(C)\le 2\log n$, and define $z'=z'_G(p)=\sum_{C\in\mathcal C}p^{|C|}$. Clearly, $z'_G(p)\le z_G(p)$. The next lemma shows that $z'_G(p)$ is a tight upper bound for $\zgp$ when the condition in \Cref{def:surrogate-criteria} is true. 

\begin{lemma}[Reliable Case]\label{lem:ignore-cuts-log-n}
If $n^{1+30\delta}\tilde k_{\pi_i}p^{2\pi_i}\le1$ for all $i\ge 0$, then $z'_G(p)\ge(1-O(1/n))z_G(p)$. 
\end{lemma}
\begin{proof}
Equivalently, we need to prove that $z_G(p)-z'_G(p)\le O(1/n)z_G(p)$, where $z_G(p)-z'_G(p)=\sum_{C\notin\mathcal C}p^{|C|}$. Define $\mathcal C'$ as the collection of cuts $C$ with $\ell^*(C)\le\log n$. Since $\ell^{\mathcal T}(e)\approx_\delta\ell^*(e)$ for all $e\in E$, any cut with $\ell^*(C)\le\log n$ also satisfies $\ell^{\mathcal T}(C)\le2\log n$, so $\mathcal C'\subseteq\mathcal C$. So it suffices to show that $\sum_{C\notin\mathcal C'}p^{|C|}\le O(1/n)z_G(p)$.
We have two cases depending on whether $p^\lambda > n^{-2}/\log n$.

Assume first that $p^\lambda > n^{-2}/\log n$. Consider a fixed cut $C$, and let $C_i$ be the cut edges on level $i$.
Since the min-ratio cut algorithm sets $\ell^*(e)=\frac1{\pi_i}$ for all edges $e\in C_i$, we have $|C|=\sum_i\ell^*(C_i)\pi_i$. From the assumption, we have $p^{\pi_i}\le(\max\{k_{\pi_i},\sqrt n\})^{-(1+\delta)}$ for all $i$. We bound
\begin{align}
    \sum_{C\notin\mathcal C'}p^{|C|} = \sum_{C:\ell^*(C) > \log n} p^{|C|} &=  \sum_{C:\ell^*(C) > \log n} \prod_i p^{\ell^*(C_i) \pi_i}\nonumber
    \\&\le\sum_{C:\ell^*(C) > \log n}  \prod_i (\max\{k_{\pi_i},\sqrt n\})^{-(1+\delta)\ell^*(C_i)}\label{eq:zL}
.
\end{align}
We now bound the number of cuts $C$ with $\ell^*(C)=\alpha$ for a given $\alpha$.
For each cut $C$, if we take a sample $T$ from the ideal tree packing, we obtain $\mathbb E[|C\cap T|]=\ell^*(C_i)$. By Markov's inequality, with probability at most $(1+\delta/2)^{-1}$ the sampled tree $T$ satisfies $|C_i\cap T|>(1+\delta/2)\ell^*(C_i)$. So with probability $\Omega(\delta)$, we have $|C_i\cap T|\le(1+\delta/2)\ell^*(C_i)$. We can sample a tree from the ideal tree packing such that the $k_{\pi_i}-1$ edges of the tree on each level $i$ are independent. So with probability $\Omega(\delta)^{O(1/\delta)}$, we have $|C_i\cap T|\le(1+\delta/2)\ell^*(C_i)$ for all $i$. 
This probability is for a random tree in packing. So, if the tree packing has size at least $O(1/\delta)^{O(1/\delta)}$, then there exists a tree $T$ to make it happen. It follows that for any collection of integers $\{\alpha_i\}_i$ there are at most $O(1/\delta)^{O(1/\delta)}\cdot\prod_ik_{\pi_i}^{(1+\delta/2)\alpha_i}$ many cuts $C$ with $|C_i\cap T|=\alpha_i$ for all $i$. Continuing from (\ref{eq:zL}), we union bound over all $\{\alpha_i\}_i$ summing to $\alpha\ge\log n$ to get
\begin{align*}
\sum_{C\notin\mathcal C'}p^{|C|}&\le\sum_{\alpha\ge\log n}\sum_{\sum_i\alpha_i=\alpha} O(1/\delta)^{O(1/\delta)}\cdot\prod_ik_{\pi_i}^{(1+\delta/2)\alpha_i}\cdot \prod_i (\max\{k_{\pi_i},\sqrt n\})^{-(1+\delta)\alpha_i}
\\&\le\sum_{\alpha\ge\log n}\sum_{\sum_i\alpha_i=\alpha}O(1/\delta)^{O(1/\delta)}\cdot\prod_i(\max\{k_{\pi_i},\sqrt n\})^{-\delta\alpha_i}
\\&\le\sum_{\alpha\ge\log n}\sum_{\sum_i\alpha_i=\alpha}O(1/\delta)^{O(1/\delta)}\cdot\prod_i(\sqrt n)^{-\delta\alpha_i}
\\&\le\sum_{\alpha\ge\log n}\sum_{\sum_i\alpha_i=\alpha}O(1/\delta)^{O(1/\delta)}\cdot n^{-\delta\alpha/2}.
\end{align*}
For a fixed integer $\alpha$, there are at most $\binom{\alpha}{O(1/\delta)}$ ways to write it as $\alpha=\sum_i\alpha_i$, so the expression above is at most
\[ \sum_{\alpha\ge\log n}\alpha^{O(1/\delta)}\cdot O(1/\delta)^{O(1/\delta)}\cdot n^{-\delta\alpha/2}\le(\log  n)^{O(1/\delta)}\cdot O(1/\delta)^{O(1/\delta)}\cdot n^{-\delta(\log n)/2}\le n^{o(1)-\frac{\log n}{2\log \log n}}\le O(1/n)\cdot p^\lambda, \]
concluding the case $p^\lambda > n^{-2}/\log n$.

Next, assume that $p^\lambda \le n^{-2}/\log n$. Observe that any cut $C$ with $\ell^*(C)\ge\log n$ satisfies $|C| = \sum_i |C_i| \ge \sum_i \ell(C_i)\pi_i\ge\sum_i\ell(C_i)\pi^*\ge\log n \cdot \lambda/2$, i.e., it is a $(\frac12\log n)$-approximate minimum cut.
We use $z\ge p^\lambda$ and the standard cut-counting argument to get
    \begin{align*}
    \frac1z\sum_{C\notin\mathcal C'}p^{|C|} 
    &\le p^{-\lambda}\sum_{j \ge\frac12\log n}n^{2(j+1)}p^{j\lambda}
    = n^4\sum_{j \ge\frac12\log n}n^{2(j-1)}p^{(j-1)\lambda}\\
    &= n^4\sum_{j \ge\frac12\log n}(n^2p^\lambda)^{j-1}
    \le n^4\cdot O\left(\frac{1}{\log n}\right)^{\log n/2},
    \end{align*}
    which is at most $O(1/n)$, concluding the proof.
\end{proof}

To prove \Cref{lem:interface-importance-zgp}, we are left to give an algorithm for an estimator of $z'_G(p)$ and bound its variance. We do this over the next two subsections.



\subsection{Cut Sampling Algorithm on Tree Packing}\label{sec:importance}
We now give a cut sampling algorithm on the approximate tree packing that we described above. 
The goal of this algorithm is to design an estimator for $z'_G(p)$.

First, we define a distribution $q$ on the cuts of the graph.
The algorithm will sample a cut $C$ according to this distribution $q$, i.e. each cut $C$ is sampled with probability $q(C)$, and then compute the estimator
\[ X=\begin{cases} p^{|C|}/q(C) &\text{ if }\ell^{\mathcal T}(C)\le2 \log n,\\
0&\text{ otherwise.}\end{cases}\]

Our high-level plan is to establish the following properties:
\begin{enumerate}
    \item After $\tilde{O}(m)$ preprocessing time, we can sample a cut $C$ according to $q$ and compute  $X$ in $\tilde{O}(1)$ time.
    \item The relative variance of $X$ is $n^{1+o(1)}$.
\end{enumerate}
If these hold, then by taking average of $n^{1+o(1)}$ samples, we obtain an estimator for $z'_G(p)$ with relative variance $O(1)$ in $\tilde O(m)+n^{1+o(1)}$ time by \Cref{fact:rel-var-decrease}. 
(Recall that by \Cref{lem:ignore-cuts-log-n}, this gives an estimator for $\zgp$ with relative bias $O(1/n)$, which establishes \Cref{lem:interface-importance-zgp}.)

Property 1 essentially follows from prior work~\cite{CenHLP24} (given later in \Cref{lem:data-structure}).
We focus on satisfying Property 2 first. 
In fact, we show a stronger property. Define the \emph{likelihood ratio} of a cut $C$ as $\rho(C) = \frac{p^{|C|}}{q(C)\cdot z'_G(p)}$, i.e., the ratio of the fractional contribution $p^{|C|}/z'_G(p)$ of cut $C$ to $z'_G(p)$ and the probability $q(C)$ of being sampled. Then,

\begin{enumerate}
\item[2'.] For each cut $C\in\mathcal C$, we have $\rho(C)\le n^{1+o(1)}$.
\end{enumerate}
Property 2' implies property 2 because of the following lemma.
\begin{lemma}\label{lem:relative-variance}
    Suppose $z'=\sum_{C\in \mathcal{C}} p^{|C|}$ for some collection of cuts $\mathcal{C}$, and $X$ is an importance sampling estimator such that $X=p^{|C|}/q(C)$ when $C\in \mathcal{C}$ and $X=0$ otherwise. Here, $q(C)$ is the probability of sampling $C$.
    Then, the relative variance of $X$ is upper bounded by the maximum likelihood ratio $\displaystyle\max_{C\in\mathcal{C}} \frac{p^{|C|}}{q(C)\cdot z'}$.
\end{lemma}
\begin{proof}
%
The relative variance is defined as
\begin{align*} 
\frac{\mathbb E[X^2]}{\mathbb E^2[X]}-1
&\le\frac{\mathbb E[X^2]}{\mathbb E^2[X]}
\le\frac{(\max X)\mathbb E[X]}{\mathbb E^2[X]}
=\frac{\max X}{\mathbb E[X]}\\
&=\frac{\displaystyle\max_{C\in\mathcal{C}}p^{|C|}/q(C)}{\sum_{C\in \mathcal{C}}q(C)\cdot \frac{p^{|C|}}{q(C)}}
=\max_{C\in C} \frac{p^{|C|}}{q(C)\cdot z'}.\qedhere
\end{align*}

\end{proof}



\medskip\noindent{\bf Sampling Algorithm on Approximate Tree Packing.}
%
Our goal in this section is to describe a distribution $q$ on cuts such that the likelihood ratio $\rho(C) = \frac{p^{|C|}}{q(C)\cdot z'_G(p)}$ is $n^{1+o(1)}$ for all cuts $C$ with  $\ell^{\mathcal T}(C) \le 2\log n$. Note that the cuts $C$ with $\ell^{\mathcal T}(C)>2\log n$ contribute only an $O(1/n)$ fraction of $z_G(p)$ by \Cref{lem:ignore-cuts-log-n}.

The distribution $q$ is defined by a process that samples a random cut.
%
For computational reasons, we first sparsify the tree packing $\mathcal T$ into another tree packing $\widetilde{\mathcal T}$ by uniformly sampling $\tilde O(1)$ trees in $\mathcal T$ with replacement. Using the implicit representation of $\mathcal T$, we can explicitly recover all trees in $\widetilde{\mathcal T}$ in $\tilde O(n)$ time.

Now, we uniformly sample a random tree $T$ from the sparsified packing $\widetilde{\mathcal T}$, and also uniformly sample random integers $j_1,j_2,\ldots,j_L$ conditioned on $\sum_ij_i\le 8\log n$ and $j_i\le|E(T)\cap\widetilde E_i|$ for all $i\in[L]$. Then, we uniformly sample $j_i$ edges from $E(T)\cap\widetilde E_i$ (without replacement). The sampled edges uniquely define a vertex bi-partition such that the set of edges in $T$ crossing the bi-partition is exactly the set of sampled edges. We output the cut corresponding to this bi-partition.



The next lemma bounds the likelihood ratio $\frac{p^{|C|}}{q(C)\cdot z'_G(p)}$ of this sampling procedure for the cuts involved in $z'_G(p)$, under the additional assumption that $p^\lambda\ge n^{-1000}.$\footnote{We can use any sufficiently large constant instead of $1000$ here.} In \Cref{sec:very-reliable}, we give a separate algorithm that handles the $p^\lambda<n^{-1000}$ case.

\begin{lemma}\label{lem:reliable-main}
    Assume that $n^{1+30\delta}\tilde k_{\pi_i}p^{2\pi_i}\le1$ for all $i$, and that $p^\lambda\ge n^{-1000}$.
    With high probability, for each cut $C$ with $\ell^{\mathcal T}(C) \le 2\log n$, the likelihood ratio of $C$ is $(\log n)^{O(1/\delta)}n^{1+O(\delta)}$.
\end{lemma}

\begin{proof}
We first establish a property of the sparsified $\widetilde{\mathcal T}$ that holds with high probability. We define a certain condition (namely~(\ref{eq:ji-log-ki-bound})) on each cut $C$ which we show holds with probability $\Omega(\delta)$ for a random tree $T\in\mathcal T$. 
Since the sparsified $\widetilde{\mathcal T}$ consists of $\tilde{O}(1)$ many random trees in $\mathcal T$, the probability that at least one $T\in\widetilde{\mathcal T}$ satisfies the condition is at least $1-n^{-10\log n}$. Since there are only $|\mathcal T|\cdot O(n^{8\log n})$ many possible cuts $C$ that can be returned by the sampler, by a union bound we conclude that at least one tree $T\in\mathcal T$ satisfies the condition for all cuts $C$ we need to consider.

Now consider a fixed cut $C\in\mathcal C$, and define $C_i=C\cap\widetilde E_i$. Sample a tree $T\in\widetilde{\mathcal T}$, and define $j^*_i=|C_i\cap E(T)|$ for all $i\in[L]$.

For all $i$, define $\tilde k'_{\pi_i} = \max\{\tilde k_{\pi_i}, \sqrt{n}\}$.
By \Cref{lem:tilde-ki-range}, 
\begin{gather}
\prod_i (\tilde k_{\pi_i}(T))^{j^*_i} \le \prod_i (\tilde k_{\pi_{i+2}})^{j^*_i}
\le \prod_i (\tilde k'_{\pi_{i+2}})^{j^*_i}
= \exp(\sum_i j^*_i \log \tilde k'_{\pi_{i+2}}).\label{eq:number-of-samples}
\end{gather}
For a given cut $C$, define $C_i=C\cap\widetilde E_i$. By definition, the load $\ell^{\mathcal T}(e)$ is the probability that a random tree $T\in\mathcal T$ contains edge $e$. Summing over all edges of $C_i$ and applying linearity of expectation, we obtain $\mathbb E_{T\sim\mathcal T}[j^*_i] = \ell^{\mathcal T}(C_i)$; note that this expectation is over a sample $T$ drawn from $\mathcal T$ instead of $\widetilde{\mathcal T}$. Summing $j^*_i\log \tilde k'_{\pi_{i+2}}$ over all $i$, we obtain
\[\mathop{\mathbb E}_{T\sim\mathcal T}\left[\sum_i j^*_i \log \tilde k'_{\pi_{i+2}}\right]
= \sum_i \ell^{\mathcal T}(C_i)\log \tilde k'_{\pi_{i+2}}.
\]
By Markov's inequality, at least a $\Omega(\delta)$ fraction of trees in $\mathcal T$ satisfy 
\begin{equation}\label{eq:ji-log-ki-bound}
    \sum_i j^*_i \log \tilde k'_{\pi_{i+2}} \le (1+\delta)\sum_i \ell^{\mathcal T}(C_i)\log \tilde k'_{\pi_{i+2}}.
\end{equation}
By our guarantee of $\widetilde{\mathcal T}$, at least one tree $T\in\widetilde{\mathcal T}$ also satisfies (\ref{eq:ji-log-ki-bound}).
In this case, we obtain
\[\frac 12 \log n \sum_i j^*_i \le \sum_i j^*_i \log \sqrt{n}\le \sum_i j^*_i \log \tilde k'_{\pi_{i+2}} \le (1+\delta)\sum_i \ell^{\mathcal T}(C_i) \log \tilde k'_{\pi_{i+2}} \le (1+\delta) \ell^{\mathcal T}(C) \log n,\]
which implies that $\sum_ij^*_i\le2(1+\delta)\ell^{\mathcal T}(C)$.
Since $\ell^{\mathcal T}(C)\le2 \log n$ and $\delta\le1$ by assumption, we have $\sum_i j^*_i \le8\log n$. In particular, it is possible for the algorithm to sample $j_i=j^*_i$ for all $i$. The probability that this occurs is the inverse of the number of combinations $j_1,j_2,\ldots,j_L$ satisfying $\sum_ij_i\le 8\log n$ and $j_i\le|E(T)\cap\widetilde E_i|$ for all $i\in[L]$. The first inequality alone limits the number of combinations to $(8\log n)^{O(L)}=(8\log n)^{O(1/\delta)}$, so the probability that $j_i=j^*_i$ for all $i$ is $(8\log n)^{-O(1/\delta)}$.

Conditioned on $j_i=j^*_i$ for all $i$, the cut $C$ is sampled precisely when we sample exactly the edges $C_i\cap E(T)$ from $\widetilde E_i\cap E(T)$ for each $i$. The total number of possible samples (over all $i$) is
\begin{align*}
\prod_i k_{\pi_i}^{j_i}=\prod_i k_{\pi_i}^{j^*_i}&\stackrel{(\ref{eq:number-of-samples})}\le  \exp(\sum_i j^*_i \log \tilde k'_{\pi_{i+2}})
 \\&\stackrel{(\ref{eq:ji-log-ki-bound})}\le \exp\bigg( (1+\delta)\sum_i \ell^{\mathcal T}(C_i)\log \tilde k'_{\pi_{i+2}}\bigg)
\\&=\prod_i(\tilde k'_{\pi_{i+2}})^{(1+\delta)\ell^{\mathcal T}(C_i)}
.
\end{align*}
Multiplying the probabilities together, the probability of sampling a tree $T\in\widetilde{\mathcal T}$ satisfying (\ref{eq:ji-log-ki-bound}), and then sampling $j_i=j^*_i$ for all $i$, and then sampling exactly the edges $C_i\cap E(T)$ from $\widetilde E_i\cap E(T)$ for each $i$ is at least
\[ \frac1{|\widetilde{\mathcal T}|} \cdot (8\log n)^{-O(1/\delta)} \cdot \prod_i(\tilde k'_{\pi_{i+2}})^{-(1+\delta)\ell^{\mathcal T}(C_i)} .\]
This is a lower bound for the probability $q(C)$ of sampling cut $C$. Note that since $|\widetilde{\mathcal T}|=\tilde O(1)$, the $1/|\widetilde{\mathcal T}|$ term can be absorbed into $(8\log n)^{-O(1/\delta)}$.

Finally, we bound the likelihood ratio 
\[\rho(C) = \frac{p^{|C|}}{q(C) \cdot z'}
    \le \frac{1}{z'} \cdot p^{|C|}\cdot(8\log n)^{O(1/\delta)}\cdot
    \prod_i (\tilde k'_{\pi_{i+2}})^{(1+\delta) \ell^{\mathcal T}(C_i)} .\]
Notice that
\[|C| = \sum_i |C_i| = \sum_i\sum_{e\in C_i}1\ge\sum_i\sum_{e\in C_i}\ell^{\mathcal T}(e)\pi_{i-1}
= \sum_i \ell^{\mathcal T}(C_i) \pi_{i-1} ,\]
so
\[\rho(C) \le \frac1{z'}\cdot\prod_ip^{\ell^{\mathcal T}(C_i)\pi_{i-1}}\cdot (8\log n)^{O(1/\delta)}\cdot
    \prod_i (\tilde k'_{\pi_{i+2}})^{(1+\delta) \ell^{\mathcal T}(C_i)} .\]

Recall from \Cref{lem:ugp-ktau} that $z\ge k_{\pi_{i+3}}p^{2\pi_{i+3}}$ for all $i$. We have $k_{\pi_{i+3}}\ge\tilde k_{\pi_{i+2}}$ by \Cref{lem:tilde-ki-range}, so we can write
\[ \frac1{z}=\prod_iz^{-\ell^{\mathcal T}(C_i)/\ell^{\mathcal T}(C)}\le\prod_i\big(\tilde k_{\pi_{i+2}}p^{2\pi_{i+3}}\big)^{-\ell^{\mathcal T}(C_i)/\ell^{\mathcal T}(C)}=\prod_i\big(\tilde k_{\pi_{i+2}}p^{2(1+\delta)\pi_{i+2}}\big)^{-\ell^{\mathcal T}(C_i)/\ell^{\mathcal T}(C)}, \]
and
\begin{align}
    \rho(C) &\le (8\log n)^{O(1/\delta)}\cdot \prod_i \left((\tilde k_{\pi_{i+2}}p^{2(1+\delta)\pi_{i+2}})^{-1/\ell^{\mathcal T}(C)}\cdot p^{\pi_{i-1}}\cdot(\tilde k'_{\pi_{i+2}})^{1+\delta}\right)^{\ell^{\mathcal T}(C_i)}\nonumber\\
    &=(8\log n)^{O(1/\delta)}\cdot \prod_i\left((\tilde k_{\pi_{i+2}})^{-1/\ell^{\mathcal T}(C)}\cdot p^{\pi_{i+2}\left(\frac1{(1+\delta)^3} - \frac{2(1+\delta)}{\ell^{\mathcal T}(C)}\right)}\cdot(\tilde k'_{\pi_{i+2}})^{1+\delta}\right)^{\ell^{\mathcal T}(C_i)}.\label{eq:rho-start}
\end{align}
Recall our assumption $n^{1+30\delta}\tilde k_{\pi_i}p^{2\pi_i}\le1$ for all $i$.
Notice that
\begin{align*}
(\tilde k'_{\pi_{i+2}})^{1+15\delta}=\max\{\tilde k_{\pi_{i+2}}, \sqrt{n}\}^{1+15\delta}
&\le \left(\sqrt{n\tilde k_{\pi_{i+2}}}\right)^{1+15\delta}
= \sqrt{n^{1+15\delta}(k_{\pi_{i+2}})^{1+15\delta}}
\le \sqrt{n^{1+30\delta} k_{\pi_{i+2}}} \le p^{-\pi_{i+2}}.
\end{align*}
So, $p^{\pi_{i+2}}\le (\tilde k'_{\pi_{i+2}})^{-(1+15\delta)}$.
Let us first assume that $\ell^{\mathcal T}(C) \ge 2-\delta$, so the exponent $\frac1{(1+\delta)^3} - \frac{2(1+\delta)}{\ell^{\mathcal T}(C)}$ of $p^{\pi_{i+2}}$ is, for large enough constant $K>0$,
\[ \frac1{(1+\delta)^3} - \frac{2(1+\delta)}{\ell^{\mathcal T}(C)} = \frac1{(1+\delta)^3} -\frac{2(1+\delta)-K\delta}{\ell^{\mathcal T}(C)}-\frac{K\delta}{\ell^{\mathcal T}(C)} \ge \frac1{(1+\delta)^3}-\frac{2(1+\delta)-K\delta}{2-\delta}-\frac{K\delta}{\ell^{\mathcal T}(C)}\ge-\frac{K\delta}{\ell^{\mathcal T}(C)}. \]
We write
\begin{align*}
 p^{\pi_{i+2}\left(\frac1{(1+\delta)^3} - \frac{2(1+\delta)}{\ell^{\mathcal T}(C)}\right)}&=p^{\pi_{i+2}\left(\frac1{(1+\delta)^3} - \frac{2(1+\delta)}{\ell^{\mathcal T}(C)}+\frac{K\delta}{\ell^{\mathcal T}(C)}\right)}p^{\pi_{i+2}(-\frac{K\delta}{\ell^{\mathcal T}(C)})} 
\\&\le  (\tilde k'_{\pi_{i+2}})^{-(1+15\delta)\left(\frac1{(1+\delta)^3} - \frac{2(1+\delta)}{\ell^{\mathcal T}(C)}+\frac{K\delta}{\ell^{\mathcal T}(C)}\right)}p^{\pi_{i+2}(-\frac{K\delta}{\ell^{\mathcal T}(C)})}
\\&=  (\tilde k'_{\pi_{i+2}})^{-\frac{1+15\delta}{(1+\delta)^3} + \frac{2+O(\delta)}{\ell^{\mathcal T}(C)}-\frac{K\delta}{\ell^{\mathcal T}(C)}}p^{\pi_{i+2}(-\frac{K\delta}{\ell^{\mathcal T}(C)})}
\\&\le  (\tilde k'_{\pi_{i+2}})^{-(1+\delta)+ \frac{2+O(\delta)}{\ell^{\mathcal T}(C)}-\frac{K\delta}{\ell^{\mathcal T}(C)}}p^{\pi_{i+2}(-\frac{K\delta}{\ell^{\mathcal T}(C)})} ,
\end{align*}
where the first inequality uses the fact that the new exponent $\frac1{(1+\delta)^3} - \frac{2(1+\delta)}{\ell^{\mathcal T}(C)}+\frac{K\delta}{\ell^{\mathcal T}(C)}$ is nonnegative, and the second inequality uses $1+15\delta\ge(1+\delta)^4$ since $\delta\le1$. Continuing from (\ref{eq:rho-start}), we are left with
\begin{align*}
\rho(C)&\le(8\log n)^{O(1/\delta)}\cdot \prod_i \left((\tilde k_{\pi_{i+2}})^{-1/\ell^{\mathcal T}(C)} (\tilde k'_{\pi_{i+2}})^{\frac{2+O(\delta)}{\ell^{\mathcal T}(C)}}(\tilde k'_{\pi_{i+2}}p^{\pi_{i+2}})^{\frac{-K\delta}{\ell^{\mathcal T}(C)}}\right)^{\ell^{\mathcal T}(C_i)}
\\&=(8\log n)^{O(1/\delta)}\cdot \prod_i \left((\tilde k_{\pi_{i+2}})^{-1} (\tilde k'_{\pi_{i+2}})^{2+O(\delta)}\right)^{\frac{\ell^{\mathcal T}(C_i)}{\ell^{\mathcal T}(C)}}\cdot\prod_i\left(\tilde k'_{\pi_{i+2}}p^{\pi_{i+2}}\right)^{-K\delta\cdot\frac{\ell^{\mathcal T}(C_i)}{\ell^{\mathcal T}(C)}} .
\end{align*}
For the first product, we use the bound $\tilde k'_{\pi_{i+2}}=\max\{\tilde k_{\pi_{i+2}}, \sqrt{n}\}\le\sqrt{n\tilde k_{\pi_{i+2}}}$ to get
\begin{align*}
\prod_i \left((\tilde k_{\pi_{i+2}})^{-1} (\tilde k'_{\pi_{i+2}})^{2+O(\delta)}\right)^{\frac{\ell^{\mathcal T}(C_i)}{\ell^{\mathcal T}(C)}}&\le\prod_i \left((\tilde k_{\pi_{i+2}})^{-1} \big(\sqrt{n\tilde k_{\pi_{i+2}}}\big)^{2+O(\delta)}\right)^{\frac{\ell^{\mathcal T}(C_i)}{\ell^{\mathcal T}(C)}}
\\&= \prod_i \left( (\tilde k_{\pi_{i+2}})^{O(\delta)} n^{1+O(\delta)}\right)^{\frac{\ell^{\mathcal T}(C_i)}{\ell^{\mathcal T}(C)}}
\\&\le\prod_i\left(n^{1+O(\delta)}\right)^{\frac{\ell^{\mathcal T}(C_i)}{\ell^{\mathcal T}(C)}}
\\&=n^{1+O(\delta)} .
\end{align*}
For the second product, we use the fact that $\pi_i\in\Theta(\lambda)$ and our assumption $p^\lambda\ge n^{-1000}$ to obtain
\[ \prod_i(\tilde k'_{\pi_{i+2}}p^{\pi_{i+2}})^{-K\delta\cdot\frac{\ell^{\mathcal T}(C_i)}{\ell^{\mathcal T}(C)}}\le\prod_i( n^{-O(1)})^{-K\delta\cdot\frac{\ell^{\mathcal T}(C_i)}{\ell^{\mathcal T}(C)}}=(n^{-O(1)})^{-K\delta}=n^{O(\delta)} .\]
We conclude that $\rho(C)\le(\log n)^{O(1/\delta)}n^{1+O(\delta)}$ under the additional assumption that $\ell^{\mathcal T}(C) \ge 2-\delta$.

To handle the case $\ell^{\mathcal T}(C)<2-\delta$, observe that by Markov's inequality, at least a $\Omega(\delta)$ fraction of trees $T\in\mathcal T$ satisfy $|C\cap E(T)|<2$, which means $|C\cap E(T)|=1$. We can similarly argue, as in the beginning of the proof, that with high probability at least one tree $T\in\widetilde{\mathcal T}$ satisfies $|C\cap E(T)|=1$. Now suppose that the single edge in $C\cap E(T)$ is in $\widetilde E_i$. Then, with probability $(8\log n)^{-O(1/\delta)}$ we have $j_i=1$ and $j_{i'}=0$ for all $i'\ne i$, and with probability $1/|\widetilde E_i\cap E(T)|\ge1/n$. So overall, the cut $C$ is sampled with probability $1/|\widetilde{\mathcal T}|\cdot(8\log n)^{-O(1/\delta)}\cdot1/n$. It follows that
\[ \rho(C)=\frac{p^{|C|}}{q(C)\cdot z'_G(p)}\le\frac1{q(C)}=(8\log n)^{O(1/\delta)}\cdot n ,\]
concluding the proof.
\end{proof}

\medskip\noindent{\bf Running Time.} 
%
%
Recall that we set $\delta=\Theta(1/\log\log n)$, so for any cut $C$ with $2/(1+\delta)< \ell^*(C) \le \log n$, the likelihood ratio $\rho(C)$ is $n^{1+o(1)}$ by \Cref{lem:reliable-main}. By \Cref{lem:relative-variance}, the relative variance of $X=p^{|C|}/q(C)$ is $n^{1+o(1)}$, so the algorithm takes $n^{1+o(1)}$ samples, 
obtaining an estimator for $\zgp$ with relative variance $O(1)$ and relative bias $O(1/n)$.

To obtain the $n^{1+o(1)}$ samples efficiently, the algorithm uses the following lemma from \cite{CenHLP24}:
\begin{lemma}[Lemma 3.8 of \cite{CenHLP24}] \label{lem:data-structure}
    Given an unweighted undirected graph $G=(V,E)$ and a tree $T$ spanning $V$, we can build a data structure $\mathbb{S}_{G,T}$ in $m^{1+o(1)}$ time such that the following query can be answered in $\tilde{O}(1)$ time: 
    Given $\tilde{O}(1)$ edges on $T$ to define a cut $C$, return its cut value on $G$.
\end{lemma}

In particular, the algorithm initializes data structure $\mathbb{S}_{G,T}$ for each $T\in\widetilde{\mathcal T}$ and then, for $n^{1+o(1)}$ iterations, performs the following: sample $T\in\widetilde{\mathcal T}$ and $j_1,\ldots,j_L$ and $j_i\le8\log n$ edges from $E(T)\cap\widetilde E_i$ for each $i$, and then query the cut value using data structure $\mathbb{S}_{G,T}$. To compute $q(C)$, the algorithm needs to figure out the probability that cut $C$ is sampled. This probability can be determined by simple combinatorics after computing $|C\cap\widetilde E_i\cap E(T')|$ for all trees $T'\in\widetilde{\mathcal T}$. To compute these values in $\tilde O(1)$ time per sample, initialize data structure $\mathbb{S}_{\widetilde E_i\cap T_1,T_2}$ for each pair of trees $T_1,T_2\in\widetilde{\mathcal T}$ and $i\in[L]$. Then $|C\cap\widetilde E_i\cap E(T')|$ equals the output of data structure $\mathbb{S}_{\widetilde E_i\cap T',T}$ when queried the sampled edges of $T$. The overall running time is $m^{1+o(1)}$. 

This concludes the analysis of the algorithm and establishes \Cref{lem:interface-importance-zgp} for $p^\lambda \ge n^{-1000}$.

\subsection{Very Reliable Case}\label{sec:very-reliable}
This section handles the very reliable case $p^\lambda < n^{-1000}$. We start by proving an analog of \Cref{lem:ignore-cuts-log-n} -- note that we are defining $z'_G(p)$ differently from \Cref{lem:ignore-cuts-log-n} for the very reliable case.

\begin{lemma}\label{lem:very-reliable-bias}
    Assume that $p^\lambda < n^{-1000}$. Let $z'_G(p) = \sum_{C:|C|<1.1\lambda} p^{|C|}$. Then $z'_G(p) \ge (1-O(1/n))\, z_G(p)$.
\end{lemma}
\begin{proof}
By the standard cut-counting bound, we have
    \begin{align*}
    \frac1z\sum_{|C|\ge 1.1\lambda}p^{|C|} &= \frac1z\bigg(\sum_{1.1\lambda\le|C|\le2\lambda}p^{|C|}+\sum_{|C|>2\lambda}p^{|C|}\bigg)\\
    &\le p^{-\lambda}\bigg(n^4p^{1.1\lambda}+\sum_{j\ge2}n^{2(j+1)}p^{j\lambda}\bigg)\\
    &\le n^4 p^{0.1\lambda} + n^4\sum_{j\ge 2} n^{2(j-1)} p^{(j-1)\lambda}\\
    &\le n^4p^{0.1\lambda}+n^4\sum_{j\ge2}(n^2p^\lambda)^{j-1},
    \end{align*}
which is $O(1/n)$ since $p^\lambda < n^{-1000}$.
\end{proof}

As in \Cref{sec:importance}, the algorithm first samples a cut $C$ according to some distribution $q$, i.e., each cut $C$ is sampled with probability $q(C)$, and then computes the estimator
\[ X=\begin{cases} p^{|C|}/q(C) &\text{ if }|C|<1.1\lambda,\\
0&\text{ otherwise.}\end{cases}\]
By the same argument as in \Cref{sec:importance}, \Cref{lem:very-reliable-bias} implies that the estimator $X$ has relative bias $O(1/n)$ to $\zgp$. It suffices to bound the likelihood ratio $\rho(C) = \frac{p^{|C|}}{q(C)\cdot z'_G(p)}$ over all cuts $C$ with $|C|<1.1\lambda$.
These cuts have $\ell^*(C) \le |C|/\pi^* \le 2|C|/\lambda < 2.2$.

\medskip\noindent{\bf Algorithm Description.}
First, we compute a Gomory-Hu tree~\cite{gomory1961multi} of $G$, which is defined as follows: 
\begin{defn}
    For an undirected graph $G = (V, E)$, a tree $T$ on the set of vertices $V$ is said to be a Gomory-Hu tree of $G$ if for every pair of vertices $u, v\in V$, there is an $(u, v)$-mincut in $T$ that has the same bi-partition and the same cut value as an $(u, v)$-mincut in $G$.
\end{defn}
Being a tree, each edge $(u, v)$ on the Gomory-Hu tree defines a $(u, v)$ min-cut; call these Gomory-Hu tree cuts. There are $n-1$ distinct Gomory-Hu tree cuts. We order their cut values in increasing order to form a list $L=(c_1, c_2, \ldots, c_{n-1})$. Also define $c_n=\max\{c_{n-1},1.1\lambda\}$.

Next, the algorithm partitions the list into $\log \log n$ levels as follows. Let $k_i = \lfloor n^{1-2^{-i}}\rfloor$ for $i=1, 2, \ldots, \log \log n$. Also define $k_0=1$ and $k_{\log \log n +1}=n$.
For each level $i=1, 2, \ldots, \log \log n+1$, the algorithm contracts all $c_{k_i}$-connected components in $G$ to form a subgraph $G_i$. This can be achieved by contracting all edges on the Gomory-Hu tree that correspond to cuts of value $\ge c_{k_i}$.
The algorithm computes a greedy tree packing $\mathcal T_i$ of $\tilde O(\lambda/\eps^2)$ trees on each contracted graph $G_i$, and then samples $\tilde{O}(1)$ of them to form $\widetilde{\mathcal T}_i$. Let $\ell_i$ be the load of $\mathcal T_i$ and let $\ell_i^*$ be the ideal load of $G_i$, so that $\ell_i\approx_{0.01}\ell^*_i$.

The distribution $q$ is defined by the following sampling procedure: First, 
we sample $i$ uniformly from 
$ \{1, 2, \ldots, \log\log n+1\}$. Then, we uniformly sample a tree $T\in\widetilde{\mathcal T}_i$ and two edges in $T$ (with replacement) to define a cut $C$. 

We bound the likelihood ratio in the next lemma:
\begin{lemma}\label{lem:very-reliable-likelihood-ratio}
Assume $p^\lambda < n^{-1000}$. With high probability, for each cut $C$ with $|C|<1.1\lambda$, the likelihood ratio of $C$  is $\tilde{O}(n)$.
\end{lemma}
\begin{proof}
    We start with the following observation:
    \begin{lemma}[Lemma 5 of \cite{Thorup08}]\label{lem:ideal-loads-contract-monotone}
    If we contract edges, then the ideal loads of remaining edges cannot increase.
    \end{lemma}
    Since $G_i$ is formed by contraction from $G$, by \Cref{lem:ideal-loads-contract-monotone}, the ideal loads $\ell^*_i$ in $G_i$ satisfy $\ell^*_i(e)\le \ell^*(e)$ for $e\in E(G_i)$.
    Since $\ell_i \approx_{0.01}\ell^*_i$, we have $\ell_i(e)\le 1.01\ell^*_i(e) \le 1.01\ell^*(e)$ for all $e\in E(G_i)$.
    We conclude the following claim:

    \begin{claim}\label{claim:reliable-case-greedy-load-approx}
        The following holds with high probability: $\ell_i(e)\le 1.01\ell^*(e)$ for each $i\in\{1,2,\ldots,\log\log n+1\}$ and each $e\in E(G_i)$.
    \end{claim}
    
    This claim is now used to establish the next claim:
    
    \begin{claim}
        For each $i\in\{1,2,\ldots,\log\log n+1\}$, if $|C|\in [c_{k_i}, c_{k_{i+1}})$ and $|C|<1.1\lambda$, then $\frac{1}{q(C)} \le \tilde{O}(1) \cdot k_{i+1}^2$.
    \end{claim}
    \begin{proof}
        Because $|C|<c_{k_{i+1}}$, $C$ does not separate any $c_{k_{i+1}}$-connected components of $G$. Because $G_{i+1}$ is formed by contracting $c_{k_{i+1}}$-connected components, $C$ is preserved in $G_{i+1}$. (In the case $i=\log\log n$, note that $G_{\log \log n+1}$ is the original graph.)
    
        The above argument also implies that the min-cut of $G$ is preserved in $G_{i+1}$. Combined with the fact that contraction cannot decrease min-cut value, we have that the min-cut value in $G_{i+1}$ is still $\lambda$. Thus, the ideal tree packing value $\pi^*_{i+1} > \frac{2}{\lambda}$ by \Cref{fact:pi-range}.
        Because $|C|< 1.1\lambda$, we have $\ell^*_{i+1}(C) \le |C|/\pi^*_{i+1} < 2|C|/\lambda < 2.2$.
        By \Cref{claim:reliable-case-greedy-load-approx}, in greedy tree packing $\mathcal T_{i+1}$, $\ell_{i+1}(C)\le 1.01\ell^*_{i+1}(C)< 2.3$. Notice that $\ell_{i+1}(C)=\mathbb E[|T\cap C|]$ for a uniformly random tree $T$ in $\mathcal T_{i+1}$. By Markov's inequality, $\Pr[|T\cap C|\ge 3] < \frac{2.3}{3}<0.8$, which implies $\Pr[|T\cap C|\le 2]>0.2$. Since $\widetilde{\mathcal T}_{i+1}$ is a random sample of $\tilde{O}(1)$ trees in $\mathcal T_{i+1}$, with high probability there exists some $T\in\widetilde{\mathcal T}_{i+1}$ such that $|T\cap C|\le 2$. Moreover, this holds with high probability over all cuts $C$ with $|C|<1.1\lambda$ since there are only $n^{O(1)}$ such cuts.
    
        Given a tree $T\in\widetilde{\mathcal T}_{i+1}$ with $|T\cap C|\le2$, the cut $C$ can be sampled if the one or two edges in $T\cap C$ are sampled by the algorithm. This happens with probability $\frac{1}{k_{i+1}^2}$. So overall, $q(C)\ge\frac1{\log\log n+1}\cdot\frac1{|\widetilde{\mathcal T}_{i+1}|} \cdot \frac{1}{k_{i+1}^2}$, where the first and second factor come from sampling the correct $i$ and $T\in\widetilde{\mathcal T}_{i+1}$. Rearranging and using $|\widetilde{\mathcal T}_{i+1}|=\tilde{O}(1)$ completes the proof of the claim. 
    \end{proof}
    
    Finally, we need the following claim:
    
    \begin{claim}
        $z'_G(p) \ge k_i p^{c_{k_i}}$ for all $i\in\{1,2,\ldots,\log\log n+1\}$. 
    \end{claim}
    \begin{proof}
        In the list of Gomory-Hu tree cuts, $c_{k_i}$ is the $k_i$-th smallest cut, so there are $k_i$ distinct cuts of value $\le c_{k_i}$. Their contribution to $z'_G(p)$ is at least $k_i p^{c_{k_i}}$.
    \end{proof}
    
    To prove \Cref{lem:very-reliable-likelihood-ratio},
    we bound the likelihood ratio for all cuts $C$ with $|C|<1.1\lambda$ by combining the above claims.
    Suppose $|C|\in[c_{k_i}, c_{k_{i+1}})$. We have
    \begin{align*}
        \frac{p^{|C|}}{q(C)\cdot z'_G(p)}\le p^{c_{k_i}}\cdot \frac{1}{q(C)} \cdot \frac{1}{k_i p^{c_{k_i}}}
        \le \frac{\tilde{O}(1)\cdot k_{i+1}^2}{k_i} = \tilde{O}(1)\cdot n^{2(1-2^{-i-1}) - (1-2^{-i})} = \tilde{O}(n).
    \end{align*}
    We have covered the whole range $[\lambda, 1.1\lambda)$ because $c_{k_0}=c_1=\lambda$ and $c_{k_{\log \log n+1}}=c_n = \max\{1.1\lambda,c_{n-1}\}\ge1.1\lambda$. This completes the proof of \Cref{lem:very-reliable-likelihood-ratio}.
\end{proof}

We now use the connection between likelihood ratio and relative variance given by \Cref{lem:relative-variance} to establish the desired bound on the relative variance  of the estimator for $z'_G(p)$.

\begin{corollary}\label{cor:relative-variance-veryreliable}
    The relative variance in the very reliable case of the estimator of $z'_G(p)$ is at most $\tilde{O}(n)$.
\end{corollary}

\medskip\noindent{\bf Running Time.}
For the sampled cut $C$, we can compute its overall sampling probability $q(C)$ using the data structure of \Cref{lem:data-structure} similar to the previous section.
We need an additional preprocessing step to construct a Gomory-Hu tree, which takes $m^{1+o(1)}$ time~\cite{ALPS23}. In conclusion, the same running time bound of $m^{1+o(1)}$ continues to hold for the very reliable case.

Now, we have established bounds on the relative variance and running time of estimators of $z'_G(p)$, as defined differently in the two cases -- reliable and very reliable. This concludes our proof of \Cref{lem:interface-importance-zgp}.

\subsection{Using Bootstrapping to Improve Relative Bias}\label{sec: bias-z-to-u}

We have established an importance sampling algorithm which outputs an estimator of $\zgp$ with relative bias $O(1/n)$ and relative variance $\le 1$.
By \Cref{lem:z-approx-u}, this is also an estimator of $\ugp$ with relative bias $O(1/\log n)$. 
But, this is not sufficient for establishing \Cref{lem:interface-importance}.
It turns out we can improve this bound by viewing the importance sampling algorithm as a bootstrapping step.

\begin{defn}\label{def:surrogate-criteria-2}
Assume we are in the reliable case of \Cref{def:surrogate-criteria}. 
Run the importance sampling algorithm to get an estimator for $\zgp$ with relative bias $O(1/n)$ and relative variance $\le 1$ guaranteed by \Cref{lem:interface-importance-zgp}. By sampling from this estimator and applying \Cref{lem:mc-sample}, whp we can obtain a 2-approximation of $z_G(p)$, denoted $\tilde{z}$.
The algorithm is still in the reliable case if $\tilde{z} < \frac{1}{2n}$. If $\tilde{z} \ge \frac{1}{2n}$, we switch to the unreliable case.
\end{defn}

If we switch to the unreliable case according to the above criteria, we merely view the importance sampling algorithm as a preprocessing step to distinguish reliable and unreliable cases.
If we are still in the reliable case, we run the importance sampling algorithm again to obtain an estimator, which serves as the output of the algorithm.
Since we have a stronger assumption than our surrogate assumption, we can exploit it to guarantee a stronger bound for relative bias.

\begin{lemma}\label{lem:criteria-2-inverse}
    If $\tilde{z}\ge\frac{1}{2n}$, $p<\theta$ and $\tilde{z}$ is a 2-approximation of $\zgp$, then $\ugp\ge n^{-1-O(1/\log \log n)}$.
\end{lemma}
\begin{proof}
    We have $\zgp \ge \frac{\tilde{z}}{2} \ge \frac{1}{4n}$. Since $p<\theta$, by \Cref{lem:z-approx-u}, $\ugp \ge \frac{1}{2}\zgp \ge n^{-1-O(1/\log \log n)}$.
\end{proof}

\begin{lemma}\label{lem:stronger-relative-bias}
If $\zgp \le n^{-1}$, then $\xgp \le n^{-\Omega(1)}\zgp$.
\end{lemma}
\begin{proof}
    We will use the following lemmas from \cite{Karger20}:
    \begin{lemma}[Theorem 5.1 of \cite{Karger20}]\label{lem:karger-x-to-z-detail-bound}
        For any $p<\theta$, $x_G(p)\le O(\log n) \cdot (p/b)^{\lambda/2} z_G(p)$. Here, $\theta$ is defined in \Cref{lem:z-approx-u}, and $b$ is defined such that $u(b)=\frac 12$.
    \end{lemma}
    
    \begin{lemma}[Theorem 7.1 of \cite{Karger20}]\label{lem:karger-partial-z-bound}
    For any $\alpha \ge 1$, let $z_{G,\alpha}(p)$ be the expected number of failed cuts whose value is at most $\alpha\lambda$, i.e., $z_{G,\alpha}(p) = \sum_{C: |C|\le \alpha\lambda} p^{|C|}$. Let $b$ be such that  $u(b)=\frac 12$. Then,
    $$z_{G, \alpha}(b) \ge n^{-O(1/\log\alpha)}.$$
    \end{lemma}

    Since $\zgp\le 1/n$, we have $p<\theta$. By \Cref{lem:karger-x-to-z-detail-bound}, it suffices to show that $(p/b)^\lambda = n^{-\Omega(1)}$.
    Because $p<b$, we have
    \[z_{G, \alpha}(p) = \sum_{C: |C|\le \alpha\lambda} p^{|C|}
    = \sum_{C: |C|\le \alpha\lambda} \left(\frac{p}{b}\right)^{|C|}b^{|C|}
    \ge \left(\frac{p}{b}\right)^{\alpha\lambda}z_{G, \alpha}(b)\]
    So,
    \[\left(\frac{p}{b}\right)^{\lambda}\le \left(\frac{z_{G, \alpha}(p)}{z_{G,\alpha}(b)}\right)^{1/\alpha}\le n^{-\frac{1}{\alpha} + \frac{C}{\alpha\log \alpha}}\]
    where $C$ is the hidden constant in  \Cref{lem:karger-partial-z-bound}.
    Choose $\alpha = e^{2C}$, then $(p/b)^\lambda = n^{-1/2e^{2C}} = n^{-\Omega(1)}$ as desired.
\end{proof}

\begin{lemma}\label{lem:importance-sampling-bias-stronger}
    If $\tilde{z} < \frac{1}{2n}$ and $z$ is a 2-approximation of $\zgp$, then the estimator output by importance sampling algorithm has $n^{-\Omega(1)}$ relative bias to $\ugp$.
\end{lemma}
\begin{proof}
    For an importance sampling estimator, its expectation $z'$ satisfies $(1-O(1/n))\zgp \le z' \le \zgp$ by \Cref{lem:ignore-cuts-log-n,lem:very-reliable-bias}.
    We have $\zgp \le 2\tilde{z} < n^{-1}$.
    Using the inclusion-exclusion bound $\zgp - \xgp \le \ugp \le \zgp$ and \Cref{lem:stronger-relative-bias}, we have $(1-O(1/n))\ugp \le z'\le (1+n^{-\Omega(1)})\ugp$.
    In conclusion the importance sampling estimator has relative bias $n^{-\Omega(1)}$ to $\ugp$.
\end{proof}

Finally, we combine the results to prove \Cref{lem:interface-importance}.
By \Cref{lem:criteria-inverse,lem:criteria-2-inverse}, we only switch to unreliable case when the assumption of \Cref{lem:interface-importance} does not hold. The algorithm outputs an importance sampling estimator in the reliable case. Its relative bias is $n^{-\Omega(1)}$ by \Cref{lem:importance-sampling-bias-stronger}.
Its relative variance is $\le 1$ by \Cref{lem:relative-variance,lem:reliable-main,lem:very-reliable-likelihood-ratio}. (The estimator takes average of $n^{1+o(1)}$ samples.)
This argument requires the boostrapping step to correctly obtain a 2-approximation, which happens whp.

\section{The Unreliable Case: Recursive Contraction}
\label{sec:contraction}
\subsection{Description of the Algorithm}

The algorithm is recursive. First, we recall the base cases from \Cref{sec:organization}.
\begin{enumerate}
    \item If $n \le \eps^{-O(1)}$, then we run Karger's algorithm \cite{Karger20} that gives an unbiased estimator for $\ugp$. 
    \item The second base case is to run Monte Carlo sampling when $p \ge \theta$ for some threshold $\theta$ whose value is given by \Cref{lem:z-approx-u}. 
    \item The final base case is invoked when \Cref{def:surrogate-criteria} and \Cref{def:surrogate-criteria-2} both hold.
    In this case, we run the importance sampling algorithm in \Cref{sec:packing}.
\end{enumerate}


We say the algorithm in the recursive case if we are not in any base case. By \Cref{lem:criteria-inverse,lem:criteria-2-inverse,lem:z-approx-u}, we have the following fact, which shows that the recursive case is a moderately unreliable case.
\begin{fact}\label{fact:recursive-case-ugp-range}
    In the recursive case, we have $n^{-1-O(1/\log \log n)} \le \ugp \le n^{-\Omega(1/\log \log n)}$.
\end{fact}

The random contraction algorithm is defined as follows: 
We say $H \sim G(q)$ if $H$ is generated from $G$ by contracting each edge independently with probability $1-q$. (We will define the value of $q$ shortly.)
We repeat this process twice independently to generate independent samples $H_1, H_2 \sim G(q)$, and recursively obtain estimators $X_1, X_2$ for $u_{H_1}(p/q)$ and $u_{H_2}(p/q)$ respectively. We return the average of $X_1$ and $X_2$ as the estimator for $\ugp$.

In the random contraction algorithm above, we use $q = 2^{-1/\gamma}$, where $\gamma$ is defined next. This parameter is related to the ideal tree packing from \Cref{sec:packing}.  
Recall that $\ell^*(e)$ is the load on edge $e$ in an ideal tree packing.
Let $\gamma^*$ be the maximum threshold such that contracting edges $e$ with $\ell^*(e)\le 1/\gamma^*$ will result in at most $\delta n$ nodes for $\delta = \Theta(1/\log\log n)$. 
%
Ideally, we would like to set $\gamma = \gamma^*$, but we don't know the ideal tree packing. Instead, we $(1\pm\delta)$-approximate the load on every edge by a greedy tree packing as described in \Cref{def:greedy-packing}. Call the loads induced by greedy tree packing as greedy loads. Now, we take the maximum threshold $\gamma'$ such that contracting edges with greedy load $\le 1/\gamma'$ produces at most $\delta n$ nodes. We set $\gamma = \gamma'$ if $\gamma' \le \lambda$, the min-cut value; else, we set $\gamma = \lambda$.

To make random contraction efficient, we want to ensure that $m \le n^{1+o(1)}$ in every recursive step. We would like to use standard cut sparsification to achieve this bound. But, sparsification introduces additional variance, which we also need to control. It turns out that the right choice is to have $m\le n^{1+\Theta\left(\frac{1}{\sqrt{\log n}}\right)}$ during the recursion. When this is violated, we perform a sparsification step given by the following lemma.
(The lemma essentially follows Section 5 of \cite{CenHLP24}, but our assumption is slightly weaker. We include a proof in the appendix for completeness.)
\begin{restatable}{lemma}{sparsify}\label{lem:sparsify-before-contract}
    Given a graph $G$ and failure probability $p$ such that the algorithm is in the recursive case, we can generate a sparsifier $H$ with $\lambda_H=\tO(1)$ and a probability $q < p$, such that $u_H(q)$ is an unbiased estimator of $\ugp$ with relative second moment at most $2+O\left(\frac{1}{\log n}\right)$.
    The running time is $\tO(m)$.
\end{restatable}

Whenever we perform sparsification, we also reduce the relative variance by branching into two recursive calls $H_1, H_2$ by making two independent calls to the sparsfication algorithm. We then recursively obtain estimators $X_1, X_2$ for $u_{H_1}(q')$ and $u_{H_2}(q')$ respectively, and return the average of $X_1$ and $X_2$. (The value of $q'$ is given by \Cref{lem:sparsify-before-contract}.)


\subsection{Bias of the Estimator}
\begin{lemma}\label{lem:contract-bias}
    The algorithm outputs an estimator of $\ugp$ with relative bias $n^{-\Omega(1)}$.
\end{lemma}
\begin{proof}
    The proof is by induction on vertex size.
    The first two base cases are unbiased. The last base case of importance sampling has relative bias $\beta=n^{-\Omega(1)}$ by \Cref{lem:interface-importance}.

    In both inductive steps of random contraction and sparsification, we output the estimator $X=\frac{X_1+X_2}{2}$, where $X_1, X_2$ are recursive estimators for $u_1, u_2$ respectively, and $\E[u_1]=\E[u_2] = \ugp$.
    By the inductive hypothesis, $|\E[X_1]-u_1|\le \beta u_1, |\E[X_2]-u_2|\le \beta u_2$.
    \begin{align*}
    |\E[X]-\ugp| 
    &= \nf 12 \cdot |\E[X_1]+\E[X_2]- \E[u_1] -\E[u_2]|\\
    &\le \nf 12 \cdot (\E_{H_1} [|\E[X_1] - u_1| \mid H_1] + \E_{H_2}[|\E[X_2]-u_2| \mid H_2]|)\\
    &\le \nf{1}{2}\cdot (\E_{H_1}[\beta u_1] + \E_{H_2}[\beta u_2]) \\
    &= \beta\cdot \ugp.\qedhere
    \end{align*}
\end{proof}

\subsection{Some Important Properties}

We are left to bound the relative variance of the estimator and the running time of the algorithm. For this, we first need to establish some properties of the parameter $\gamma$ that we used in the algorithm. Then, we argue about the rate at which the vertex size decreases when we run random contraction. This is used both in bounding relative variance and the running time of the algorithm.

\medskip\noindent{\bf Properties of $\gamma$.}
First, we give two properties of $\gamma$.
    Recall that $\delta = \Theta(1/\log\log n)$ is the parameter such that greedy tree packing approximates the load on every edge in ideal tree packing up to a factor of $1\pm \delta$ by \Cref{cor:greedy-packing}.
    
\begin{fact}\label{fact:gamma-star-to-gamma}
    $\gamma^* \in (1\pm \delta)\gamma'$.
\end{fact}
\begin{proof}
    By definition, contracting edges with $\ell^*(e) \le 1/\gamma^*$ gives $< \delta n$ nodes but contracting edges with $\ell^*(e) \le 1/\gamma^*$ gives $\ge \delta n$ nodes. Now, we note the following:
    \begin{enumerate}
        \item First, for any edge $e$, if $\ell^*(e) \le 1/\gamma^*$, then $\ell(e) \le (1+\delta)/\gamma^*$. Therefore, contracting all edges with $\ell(e) \le (1+\delta)/\gamma^*$ also yields at most $\delta n$ nodes. Hence, $\gamma' \ge \gamma^*/(1+\delta)$.
        \item Next, suppose we contract all edges $e$ with $\ell(e) < (1-\delta)/\gamma^*$. All these edges have  $\ell^*(e) < 1/\gamma^*$. Therefore, this contraction will yield $\ge \delta n$ nodes. Therefore, $\gamma' \le \gamma^*/(1-\delta)$. \qedhere
    \end{enumerate}
\end{proof}

\begin{fact} \label{fact:u-to-p-gamma}
    When $\ugp\ge n^{-1-O(1/\log \log n)}$, it holds that $\ugp \ge p^{(1+\delta)\cdot \gamma}$.
\end{fact}
\begin{proof}
There are two cases depending on whether $\gamma=\lambda$ or $\gamma = \gamma'$.
In the first case, we have $\ugp \ge p^\lambda \ge p^{(1+O(1/\log \log n))\lambda}$.

In the second case, we first note that $\ugp \ge n^{-1-O(1/\log \log n)}$ also implies $\ugp \ge (\delta n)^{-1-\Theta(1/\log \log n)}$ since
\[
\delta^{-1-\Theta(1/\log \log n)} 
\le O(\log\log n)^{1+O(1/\log \log n)}
\le n^{O(1/\log\log n)}.
\]
Next, recall that $k_{\tau}$ is the number of nodes after contracting all edges $e$ with $\ell^*(e) < 1/\tau$.
    By \Cref{lem:ugp-ktau} and \Cref{fact:gamma-star-to-gamma}, 
    \[
        \ugp\ge k_{\gamma^*} p^{2\gamma^*}\ge \delta n \cdot p^{2(1+\delta)\gamma}.
    \]
    We combine this with $\ugp\ge (\delta n)^{-1-\Theta(1/\log \log n)}$, which is equivalent to 
    \[
        (\ugp)^{1+\Theta(1/\log \log n)} \ge (\delta n)^{-1}.
    \]
    
    Multiplying the two inequalities above, we get 
    \[
        (\ugp)^{2+\Theta(1/\log \log n)} \ge   p^{2(1+\delta)\gamma}.
    \]
    Therefore, $(\ugp)^2 \ge p^{2(1+\delta)\gamma}$, or equivalently, 
    $\ugp \ge p^{(1+\delta)\gamma}$.
\end{proof}

\medskip\noindent{\bf Size Decrease.}
We show that the vertex size decreases sufficiently after a random contraction step.

\begin{lemma}\label{lem:contract-size-dec}
    Whp, it holds that for every random contraction step on an $n$-vertex graph, the vertex size after contraction is at most $\left(\frac 12 + O\left(\frac{1}{\log \log n}\right)\right)n$.
\end{lemma}
\begin{proof}
We first argue about a single random contraction step.
Consider the clusters formed if we removed all edges with $\ell^*(e) > 1/\gamma^*$.
Each cluster has min-ratio cut value $\ge \gamma^*$.
Suppose a cluster has $n_i$ vertices. We use the following claim to bound the cluster size after random contraction:

\begin{claim}\label{lem:size-dec-cluster}
    Suppose $G$ has min-ratio cut $\pi$. Then $H\sim G(q)$ has vertex size at most $1+(1+\eps)q^{\pi} n$ with probability at least $1-\exp(-\Omega(\eps^2 n))$.
\end{claim}
\begin{proof}
    Consider a continuous-time random process where each edge independently arrives at a time $t_e$, and the arrival times follow independent exponential distributions of rate 1. Let $G_t$ be the random graph formed by contracting all edges arriving up to time $t$. Let $n_t$ and $m_t$ be the number of vertices and edges in $G_t$ respectively. We consider how $n_t$ evolves during the process. Let $t(k)$ be the first time when $n_t\le k$. Initially $n_0=n$ and $t(n)=0$. Now, $n_t$ decreases by 1 whenever an uncontracted edge (i.e., an edge whose endpoints are vertices that haven't been contracted to the same node yet) arrives. Since there are $m_t$ uncontracted edges at time $t$, the earliest arrival time of an uncontracted edge after time $t$ follows an exponential distribution of rate $m_t$. 
    After any edge contraction, the nodes form a partition of $G$. Thus, $\frac{m_t}{n_t-1}\ge \pi$.

    Let $\Delta_k = t(k-1)-t(k)$. As discussed above, $\Delta_k$ follows an exponential distribution of rate $m_{t(k)}\ge \pi(k-1)$.
    Next, we couple each $\Delta_k$ with a new random variable defined by the aforementioned contraction process on a star graph. Let $R$ be a star on $n$ vertices and $n-1$ edges. We use the same notations as above for graph $R$, but with a superscript of $R$. Note that during the contraction process, we always have $m_t^R = n_t^R-1$.
    Therefore, $\Delta_k^R=t^R(k-1)-t^R(k)$ follows an exponential distribution of rate $k-1$.
    Let $\Delta'_k = \Delta_k^R / \pi$. Then, $\Delta'_k$ follows an exponential distribution of rate $(k-1)\pi$.
    Because of the memoryless property of exponential distributions and the independence of edges, $\Delta_k$ or $\Delta^R_k$ for different $k$ are independent.
    We can couple $\Delta_k$ and $\Delta'_k$ so that each $\Delta_k$ is stochastically dominated by $\Delta'_k$.
    Then,
    \[t(k)=t(k)-t(n) = \sum_{i=k+1}^n \Delta_i \le
    \sum_{i=k+1}^n \Delta'_i 
    = \sum_{i=k+1}^n \frac{\Delta^R_i}{\pi} = \frac{t^R(k)}{\pi}\]
    This implies $n_t \le n^R_{\pi t}$.

    The size of $G(q)$ is $n_{-\ln q}$, which is upper bounded by $n^R_{-\pi \ln q}$. $n^R_{-\pi \ln q}$ is the vertex size of $R$ after contracting each edge with probability $1-q^\pi$.
    So, its expectation is $1+q^{\pi}(n-1)$. 
    By Hoeffding's inequality, $\Pr[|V(H)|> 1+(1+\eps)q^{\pi} n] \le \exp(-\Omega(\eps^2 n))$.
\end{proof}

Then, by \Cref{lem:size-dec-cluster}, the size of the cluster after the random contraction step is at most $1+(1+\delta)q^{\gamma^*}n_i$. By the definition of $\gamma^*$, there are at most $\delta n$ clusters. Therefore, the number of vertices after contraction is at most $\delta n + (1+\delta)q^{\gamma^*}n$.
In the recursive contraction algorithm, 
we chose $q$ such that $q^{\gamma} = \frac 12$. We also have that
    $\gamma \le \gamma'$, which means $q^{\gamma'} \le \frac 12$. By \Cref{fact:gamma-star-to-gamma}, we have $\gamma^* \ge (1-\delta)\gamma'$.
 This implies that 
 \[
    q^{\gamma^*} \le (q^{\gamma'})^{1-\delta} \le
 (\nf 12)^{1-\delta}
 \le
 \nf 12 + O\left(\frac{1}{\log \log n}\right).
 \]
So, the number of vertices is at most $\left(\frac 12 + O\left(\frac{1}{\log \log n}\right)\right) n$ whp.

The failure probability of the bound above in a single random contraction step is $\exp(-\Omega(\delta^2 n))$ by \Cref{lem:size-dec-cluster}. So, we can apply a union bound over the polynomially many random contraction steps and conclude the statement of the lemma.
%
\end{proof}

\subsection{Relative Variance of the Estimator}

We first prove bound the relative second moment of a random contraction step.
By \Cref{lem:z-approx-u}, in the unreliable case, $\zgp$ is a sufficiently good approximation to $\ugp$, where $\zgp$ is the expected number of failed cuts. So, instead of analyzing the relative variance of $u_H(p/q)$, we analyze the relative variance of $z_H(p/q)$. We bound its second moment below:

\begin{lemma}\label{lem:relvar-1step}
The relative second moment of $z_H(p/q)$ is  $2+O\left(\frac{\log \log n}{\log n}\right)$ where $H \sim G(q)$ and $q^\gamma = \nicefrac 12$.
\end{lemma}
\begin{proof}
We abbreviate cut value $|C_i|$ by $c_i$ in the proof.
By \Cref{fact:u-to-p-gamma}, we have that $\ugp\ge p^{(1+\delta)\gamma}$.
Let $\tilde{\gamma}=(1+\delta)\gamma$ such that $\ugp\ge p^{\tilde{\gamma}}$.
Then 
\[
    q^{-\tilde{\gamma}}
    = 2^{1+\delta} 
    = 2^{1 + O(1/\log\log n)}
    = 2+O(1/\log \log n).
\]

The relative second moment is bounded as:
\begin{align*}
\frac{\E[z_H(p/q)^2]}{(\zgp)^2}
&= \frac{1}{(\zgp)^2}\cdot \E\bigg[\sum_{C_i, C_j \in \mathcal{C}(H)}\left(\frac{p}{q}\right)^{c_i+c_j}\bigg] 
= \frac{1}{(\zgp)^2}\sum_{C_i, C_j\in \mathcal{C}(G)} q^{|C_i\cup C_j|}\cdot \left(\frac{p}{q}\right)^{c_i+c_j}\\
&= \frac{1}{(\zgp)^2}\left(\sum_{C_i}\frac{p^{2c_i}}{q^{c_i}} + \sum_{C_i\ne C_j: |C_i\cap C_j|\le \tilde{\gamma}} \frac{p^{c_i+c_j}}{q^{|C_i\cap C_j|}} + \sum_{C_i\ne C_j: |C_i\cap C_j| > \tilde{\gamma}} \frac{p^{c_i+c_j}}{q^{|C_i\cap C_j|}}\right).
\end{align*}
In the above expression, we distinguished between the cases $C_i = C_j$ and $C_i \ne C_j$, and the second case is further split into $|C_i \cap C_j| \le \tilde{\gamma}$ and $|C_i\cap C_j| > \tilde{\gamma}$. We define:
\begin{align*}
    V_1 &= \sum_{C_i}\frac{p^{2c_i}}{q^{c_i}}\\
    V_2 &= \sum_{C_i\ne C_j: |C_i\cap C_j|\le \tilde{\gamma}} \frac{p^{c_i+c_j}}{q^{|C_i\cap C_j|}}\\
    V_3 &= \sum_{C_i\ne C_j: |C_i\cap C_j| > \tilde{\gamma}} \frac{p^{c_i+c_j}}{q^{|C_i\cap C_j|}}.
\end{align*}

We first bound $V_1$:

\[
V_1 - q^{-\tilde{\gamma}}\sum_{C_i} p^{2c_i}
= \sum_{C_i} p^{2c_i}\left(\frac{1}{q^{c_i}} - \frac{1}{q^{\tilde{\gamma}}}\right)
= \frac{p^{\tilde{\gamma}}}{q^{\tilde{\gamma}}} \sum_{C_i} p^{c_i} \cdot p^{c_i-\tilde{\gamma}}(q^{-(c_i-\tilde{\gamma})}-1).
\]

Let us also denote $t_i=\frac{c_i-\tilde{\gamma}}{\tilde{\gamma}}$ and $f(t)=(q^{-\tilde{\gamma}}\cdot w)^t-w^t = p^{c_i-\tilde{\gamma}}(q^{-(c_i-\tilde{\gamma})}-1)$ where $w=p^{\tilde{\gamma}}$. Therefore, 
\[
    V_1 - q^{-\tilde{\gamma}}\sum_{C_i} p^{2c_i} = \frac{p^{\tilde{\gamma}}}{q^{\tilde{\gamma}}} \sum_{C_i} p^{c_i}\cdot f(t_i).
\]
\begin{fact}
    Suppose $0 < w \le n^{-\Omega(1/\log \log n)}$ and $2 \le \alpha \le e$. Function $f(t)=(\alpha w)^t-w^t$ has global maximum $\le O\left(\frac{\log \log n}{\log n}\right)$.
\end{fact}
\begin{proof}
Let $f'(t) = (ew)^t-w^t\ge f(t)$.
    \begin{align*}
    \frac{df}{dt}
    &= (ew)^t \ln (ew) - w^t \ln w
    = e^t w^t \left(1 - (1-e^{-t})\ln\frac{1}{w}\right).
\end{align*}
$\frac{df}{dt}=0$ at $t^*=\ln\left(\frac{\log w}{\log (ew)}\right)$.
Notice that $1 - (1-e^{-t})\ln\frac{1}{w}$ is monotone decreasing, so $\frac{df}{dt}\ge 0$ when $t<t^*$ and $\frac{df}{dt}\le 0$ when $t>t^*$. Thus, $f(t^*)$ is a global maximum.

We now calculate the value of the global maximum. First, we get
\[e^{t^*}-1 = \frac{\log w}{\log (ew)}-1 = -\frac{1}{\ln (ew)}  \le O\left(\frac{\log \log n}{\log n}\right). \]
We have $w^{t^*}\le 1$ because $w\in(0,1)$ and $t^*>0$. So $f(t^*) \le w^{t^*}(e^{t^*}-1) \le O\left(\frac{\log \log n}{\log n}\right)$.
\end{proof}

Combining this fact with $p^{\tilde{\gamma}}\le \ugp\le\zgp$, $q^{-\tilde{\gamma}} = O(1)$, and $\zgp = \sum_{C_i} p^{c_i}$, we get
\begin{align}\label{eq:v1}
    V_1 - q^{-\tilde{\gamma}}\sum_{C_i}p^{2c_i}
    \le O\left(\frac{\log \log n}{\log n}\right) \cdot (\zgp)^2.
\end{align}

Next, we bound $V_2$ as follows:
\begin{align}\label{eq:v2}
V_2 
= \sum_{C_i\ne C_j, |C_i\cap C_j|\le \tilde{\gamma}} \frac{p^{c_i+c_j}}{q^{|C_i\cap C_j|}}
\le q^{-\tilde{\gamma}}\cdot \sum_{C_i, C_j} p^{c_i+c_j} - q^{-\tilde{\gamma}}\cdot \sum_{C_i} p^{2c_i}
= q^{-\tilde{\gamma}}\cdot (\zgp)^2 - q^{-\tilde{\gamma}}\cdot \sum_{C_i} p^{2c_i}.
\end{align}

Finally, we bound $V_3$.
Recall that $x_G(p)$ is the expected number of failed cut pairs in \Cref{lem:z-approx-u}.
\[
V_3 
= \sum_{C_i\ne C_j, |C_i\cap C_j|>\tilde{\gamma}} p^{|C_i\cup C_j|}\left(\frac{p}{q}\right)^{|C_i\cap C_j|}
\le \sum_{C_i\ne C_j} p^{|C_i\cup C_j|}\left(\frac{p}{q}\right)^{\tilde{\gamma}}
= q^{-\tilde{\gamma}} \cdot x_G(p)\cdot p^{\tilde{\gamma}} 
\le q^{-\tilde{\gamma}}\cdot x_G(p) \cdot \zgp.
\]
Since $p< \theta$, \Cref{lem:z-approx-u} gives $x_G(p) \le \frac{1}{\log n}\cdot \zgp$. Therefore,
\begin{align}\label{eq:v3}
    V_3 \le  q^{-\tilde{\gamma}}\cdot x_G(p) \cdot \zgp \le \frac{3}{\log n}\cdot (\zgp)^2.
\end{align}    

Putting the bounds on $V_1, V_2$ and $V_3$ given by \eqref{eq:v1}, \eqref{eq:v2}, and \eqref{eq:v3} together, we get
\[
\frac{\E[z_H(p/q)^2]}{(\zgp)^2} 
\le \frac{V_1+V_2+V_3}{(\zgp)^2}  
= 2 + O\left(\frac{\log \log n}{\log n}\right).\qedhere
\]
\end{proof}

We now use the following lemma from \cite{CenHLP24}:
\begin{lemma}[Lemma 4.6 of \cite{CenHLP24}]\label{lem:u-from-z}
Assume $p<\theta$.
If $z_{H}(p')$ is an unbiased estimator of $\zgp$ with relative variance $\eta$, and $u_{H}(p')$ is an unbiased estimator of $\ugp$, then $u_{H}(p')$ has relative variance at most $\left(1+O\left(\frac{1}{\log n}\right)\right)\eta+O\left(\frac{1}{\log n}\right)$.
\end{lemma}

Combining \Cref{lem:relvar-1step,lem:u-from-z}, we obtain the following:
\begin{corollary}\label{cor:contract-u-var}
$u_H(p/q)$ is an unbiased estimator of $u_G(p)$ with relative second moment at most $2+O\left(\frac{\log \log n}{\log n}\right)$.
\end{corollary}

We now use induction to bound the second moment of the overall estimator. This requires our bounds on the base cases as well as that established for a recursive contraction step in \Cref{lem:relvar-1step} and the corresponding bound for a sparsification step in \Cref{lem:sparsify-before-contract}.

\begin{lemma}\label{lem:contract-var}
    The second moment of the estimator given by the overall algorithm is at most $(\log n)^{O(\log \log n)} \cdot (\ugp)^2$ whp. 
\end{lemma}

\begin{proof}
Let $X$ be the estimator given by the overall algorithm.
We use induction on recursion depth to prove $\E[X^2] \le (\log n)^{K\log \log n} \cdot (u_G(p))^2$ for a large constant $K$ to be decided later.

Consider the base cases of the recursion.
In the first case (Karger's algorithm) and the second case (Monte Carlo sampling), we get an unbiased estimator of $u_G(p)$ with relative variance at most $1$ by \Cref{thm:karger-estimator} and \Cref{lem:mc}.
In the last base case, we get an estimator of $\ugp$ with relative bias at most $n^{-\Omega(1)} < 0.1$ and relative variance at most $1$ by \Cref{lem:interface-importance}.
For such an estimator $Y$, we have $\E[Y]\le 1.1\ugp$ and $\E[Y^2] - \E^2[Y]\le \E^2[Y]$. The latter can be rewritten as 
\[
    \E[Y^2] \le 2 \cdot \E^2[Y]\le 2\cdot (1.1\ugp)^2\le 3\cdot (\ugp)^2.
\]
Therefore, the statement of the lemma holds for all the base cases.

Next, we consider the inductive step where the algorithm takes the average of two recursive calls. Let $X_1$ and $X_2$ be the estimators returned by the two recursive calls; then, $X=\frac{X_1+X_2}{2}$. We have two cases depending on whether we are in a random contraction step or a sparsification step. First, we consider a random contraction step. The sparsification step is similar and discussed at the end. 

Let $H_1$ and $H_2$ be the graphs resulting from random contraction. 
Now, we have
\[\E[X^2] = \E\left[\left(\frac{X_1+X_2}{2}\right)^2\right] = \E_{H_1,H_2}\left[\frac 14 (\E[X_1^2|H_1] + \E[X_2^2|H_2]) + \frac 12 \cdot \E[X_1|H_1] \cdot \E[X_2|H_2]\right],\]
since $X_1, X_2$ are respectively independent of $H_2, H_1$.

By the inductive hypothesis, $\E[X_i^2|H_i] \le (\log n_i)^{K\log \log n} \cdot (u_{H_i}(p/q))^2$ for $i=1,2$. By applying \Cref{lem:contract-bias} on the recursive calls, we have $\E[X_i|H_i] \le 1.1u_{H_i}(p/q)$.
We can now bound the second moment of $X$ as follows:
\begin{align}
\E[X^2]
&\le  \frac 14\cdot \E\left[ (\log n_1)^{K\log \log n} (u_{H_1}(p/q))^2 + (\log n_2)^{K\log \log n} (u_{H_2}(p/q))^2\right] 
+ \frac{1}{2} \cdot \E[1.1^2u_{H_1}(p/q) u_{H_2}(p/q)]\nonumber\\
&\le  \frac 14  \cdot (\log n_1)^{K\log \log n} \cdot \E[(u_{H_1}(p/q))^2] 
+ \frac 14\cdot (\log n_2)^{K\log \log n}\cdot \E[(u_{H_2}(p/q))^2] 
+ \E[u_{H_1}(p/q)]\cdot \E[u_{H_2}(p/q)],\label{eq:xsq}
\end{align}
by independence of $H_1, H_2$. 

To bound the first two terms in \Cref{eq:xsq}, note that by \Cref{cor:contract-u-var}, we have 
\[
    \E[(u_{H_i}(p/q))^2] \le \left(2+O\left(\frac{\log \log n}{\log n}\right)\right)(u_G(p))^2 \text{ for }i=1, 2.
\]

To bound the last term in \Cref{eq:xsq}, note that $\E[u_{H_i}(p/q)] = u_G(p)$ since the random contraction step is unbiased. 
Thus, 
\[
    \E_{H_1}[u_{H_1}(p/q)]\cdot \E_{H_2}[u_{H_2}(p/q)]
    = (\ugp)^2.
\]


Putting all these bounds together, we get
\[\E[X^2]\le \frac 14 \left((\log n_1)^{K\log \log n} + (\log n_2)^{K\log \log n}\right)\left(2+O\left(\frac{\log \log n}{\log n}\right)\right) (u_G(p))^2
+  (u_G(p))^2.\]

\Cref{lem:contract-size-dec} implies that $n_i\le 0.9 n$ across all steps of recursion whp. Then, $\log n_i \le \log n - 0.1$ in all the recursion step whp. We get
\[\E[X^2] 
\le (\log n-0.1)^{K \log \log n}\left(1+O\left(\frac{\log \log n}{\log n}\right)\right) (u_G(p))^2 + (u_G(p))^2\]
Therefore, by choosing a large enough constant $K$, we ensure that
\[
\frac{\E[X^2]}{(\ugp)^2}
\le (\log n)^{K\log \log n}\left(1-\frac{0.1}{\log n}\right)^{K\log \log n}\left(1+O\left(\frac{\log\log n}{\log n}\right)\right) + 1
\le (\log n)^{K\log \log n}.
\]

For a sparsification step, we can repeat the same argument but define $H_1, H_2$ as the sparsifiers instead of the contracted graphs. By \Cref{lem:sparsify-before-contract}, $u_{H_i}(q')$ is also an unbiased estimator of $\ugp$ with relative second moment at most $2+O\left(\frac{1}{\log n}\right)$. The only caveat is that the vertex size does not decrease in the sparsification step itself. However, the vertex size decreases in the next step since it is a random contraction step. So, we can repeat the same argument by combining two steps (with relative second moment $4+O(\frac{1}{\log n})$ and 4 branches).
\end{proof}

\subsection{Running Time}

We first state a bound on the expected number of uncontracted edges after random edge contractions in an undirected graph. This was shown by Karger, Klein, and Tarjan~\cite{KargerKT95} in their celebrated randomized MST paper.
\begin{lemma}[Lemma 2.1 of \cite{KargerKT95}]\label{lem:kkt-contraction-size-bound}
Given an undirected multigraph, if we contract each edge independently with probability $\pi$, then the expected number of uncontracted edges is at most $n/\pi$. 
\end{lemma}

We now use this bound to establish the following property enforced by the sparsification steps: 

\begin{lemma}\label{lem:runtime-sparsify-step}
Assume $m > n^{1+\Omega\left(\frac{1}{\sqrt{\log n}}\right)}$ and the algorithm executes a sparsification step followed by a contraction step. Let $m'$ be the number of edges in a resulting graph. Then, the following bounds hold in expectation: $m'= \tO(n)$ and $(m')^{\frac{1}{\sqrt{\log n}}} \le \frac 12 \cdot  m^{\frac{1}{\sqrt{\log n}}}$.
\end{lemma}
\begin{proof}
Sparsification generates a graph $H$ with min-cut value $\lambda_H=\tO(1)$. Then we perform random contraction on $H$ with $q^{\gamma_H} = \nicefrac 12$. Because we defined $\gamma_H\le \lambda_H$, we have $\frac{1}{1-q} = O(\gamma_H) = O(\lambda_H)$. By \Cref{lem:kkt-contraction-size-bound}, the expected edge size of the graph after contraction is at most $\E[m'] = \frac{n}{1-q}= O(n\cdot \lambda_H) = \tO(n)$. (Note that $n$ is the vertex size before sparsification.)

Under the assumption $m > n^{1+\Omega(\frac{1}{\sqrt{\log n}})}$, we have $\E[m'] \le \tO(n) \le m^{1-\Omega\left(\frac{1}{\sqrt{\log n}}\right)}$. Notice that $x^a$ is concave when $a\in(0,1)$ and we can apply Jensen's inequality.
\[
\E\left[(m')^{\frac{1}{\sqrt{\log n}}}\right]
\le \left(\E[m']\right)^{\frac{1}{\sqrt{\log n}}}
\le m^{\left(1-\Omega\left(\frac{1}{\sqrt{\log n}}\right)\right)\frac{1}{\sqrt{\log n}}}
= m^{\frac{1}{\sqrt{\log n}}} \cdot m^{-\Omega\left(\frac{1}{\log n}\right)}
\le \nf 12 \cdot m^{\frac{1}{\sqrt{\log n}}},
\]
when the constant in the $\Omega(\cdot)$ of the assumption in the lemma's statement is sufficiently large.
\end{proof}

We are now prepared to bound the running time of the whole recursive algorithm.

\begin{lemma}\label{lem:contraction-runtime}
The algorithm runs in $m^{1+o(1)} \eps^{-O(1)}$ time in expectation.
\end{lemma}
\begin{proof}
Let $T(n, m)$ be the expected running time of a recursive call on a graph with $n$ vertices and $m$ edges. Recall that there are three types of nodes in the computation tree: base cases, contraction nodes, and sparsification nodes. If the parent of a contraction node is a sparsification node, we call it an {\em irregular} contraction node; otherwise, the contraction node is called a {\em regular} contraction node. First, we shortcut all irregular contraction nodes (by making their parent the parent of their children). The running time of all such irregular contraction nodes are accounted for by their parent sparsification nodes. Note that each sparsification node has at most two irregular contraction nodes as children; hence, it accounts for the cost of at most three nodes (itself and its two children irregular contraction nodes). Because of this transformation, the number of children of a sparsification node can increase to at most 4.

In the rest of the discussion, we assume that the contraction tree has only three types of nodes: sparsification nodes, regular contraction nodes and base cases (which are leaves of the computation tree). 
In the first base case of $n\le \eps^{-O(1)})$, the running time is $\tO(n^2)=\eps^{-O(1)}$ by \Cref{thm:karger-estimator}.
The second and third base cases take $m^{1+o(1)}$ time by \Cref{lem:interface-importance,lem:mc}.
In a sparsification node, the sparsification algorithm in \Cref{lem:sparsify-before-contract} takes $O(m)$ time.  
In a contraction node, randomly contracting edges also takes $O(m)$ time.
So, the time spent at any node of the computation tree is $m^{1+o(1)}$ in total (including the charge received by a sparsification node from its children irregular contraction nodes).

\Cref{lem:contract-size-dec} shows that whp each recursive contraction reduces the vertex size by a factor of $\frac 12 +O(\frac{1}{\log \log n})$. In particular, this holds for a regular contraction node and its children, as well as a sparsification node and its children inherited from an irregular contraction child. 

First, we consider regular contraction nodes. Note that these nodes still have at most two children.
 Moreover, they satisfy $m\le n^{1+o(1)}$ (which implies $m^{1+o(1)} \le n^{1+o(1)}$). Therefore, the recurrence is
\begin{align}\label{eq:recur-contract}
    T(n, m) \le n^{1+o(1)} + 2 \cdot T\left(\left(\frac 12 +O\left(\frac{1}{\log \log n}\right)\right)\cdot n, m\right).
\end{align}    

Now, we consider a non-root sparsification node. This is more complicated because $m > n^{1+\Omega\left(\frac{1}{\sqrt{\log n}}\right)}$. Recall that the running time incurred at this node (including that inherited from irregular contraction children) is $m^{1+o(1)}$.
Consider the parent of the sparsification node. Let $\hat{n}, \hat{m}$ respectively represent the number of vertices and edges in the parent node. If the parent is a regular contraction node, then we have $m \le \hat{m} \le \hat{n}^{1+O\left(\frac{1}{\sqrt{\log n}}\right)}$. In this case, we can charge the $m^{1+o(1)}$ term to the parent's recurence relation \Cref{eq:recur-contract}. Otherwise, the parent is a sparsification node and we have $m=\tO(\hat{n})$ in expectation by \Cref{lem:runtime-sparsify-step}. So, we can also charge $m^{1+o(1)}$ to the parent node. Finally, note that a sparsification node has at most 4 children. Let $m'$ denote the number of edges in any child of the sparsification node. We can write the following recurrence for a sparsification node:
\begin{align}\label{eq:recur-sparsify}
    T(n, m) \le n^{1+o(1)} + 4\cdot T\left(\left(\frac 12 +O\left(\frac{1}{\log \log n}\right)\right)\cdot n, m'\right), \quad \text{where $m'$ satisfies \Cref{lem:runtime-sparsify-step}.}
\end{align}

For the sake of the master theorem, we define a potential $\rho = n\log n\cdot m^{\frac{1}{\sqrt{\log n}}}$.
The progress in the first term is $n'\le \left(\frac 12 + O\left(\frac{1}{\log \log n}\right)\right)n$, so
$n'\log n'\le \left(\frac 12 + O\left(\frac{1}{\log \log n}\right)\right)n \cdot (\log n - \Omega(1)) \le \frac 12 n\log n$.
\Cref{lem:runtime-sparsify-step} measures the progress in $m^{\frac{1}{\sqrt{\log n}}}$ for the second recurrence (\Cref{eq:recur-sparsify}) by $m'= \tO(n)$ and $m'^{\frac{1}{\sqrt{\log n}}} \le \frac 12 \cdot  m^{\frac{1}{\sqrt{\log n}}}$. We have one of the following recurrence relations:
\begin{align*}
    T(\rho) &= \rho^{1+o(1)} + 2T(\nf {\rho}{2})\\
    T(\rho) &= \rho^{1+o(1)} + 4T(\nf{\rho}{4}).
\end{align*}
In either case, this solves to $T(\rho) = \rho^{1+o(1)}\eps^{-O(1)}$ (including the base cases discussed above). Therefore, 
\[
    T(n, m) = n^{1+o(1)}\eps^{-O(1)}.
\]
%
Note that the root node takes $m^{1+o(1)}$ time which isn't chargeable elsewhere. Therefore, the overall running time is $m^{1+o(1)} \eps^{-O(1)}$ in expectation.  
\end{proof}

So, we have obtained an algorithm that runs in $m^{1+o(1)}\eps^{-O(1)}$ time and obtains an estimator for $\ugp$ with relative variance $n^{o(1)}$. Overall, we repeat $n^{o(1)}\eps^{-2}$ times to get the $(1\pm \eps)$-approximation for $\ugp$ whp, thereby establishing \Cref{thm:main}. 

\section{Conclusion}
\label{sec:conclusion}
In this paper, we gave an almost-linear time algorithm to compute the unreliability of an undirected graph. Up to lower order (sub-polynomial) improvements, this brings to a close the line of work on designing fast network unreliability algorithms. However, many related problems remain open and the general area of understanding graph connectivity under random failures remains poorly understood. The complementary problem of network reliability, that estimates the probability that an undirected graph stays connected when every edge fails independently with some probability $p$, does not have close to linear time algorithms, even in the dense case when $m = \Theta(n^2)$~\cite{GuoJ19,GuoH20}. Another very interesting direction of research is to generalize the unreliability problem to more complex forms of edge failure, e.g., by allowing limited dependence between the failure events of different edges. A natural first step in this direction would be to understand the unreliability problem on hypergraphs, i.e., when all constituents edges of a hyperedge fail together. For this problem, obtaining even a polynomial-time approximation scheme remains open~\cite{CenLP24}.

\bibliographystyle{alpha}
\bibliography{ref}

\appendix
\section{Proof of \Cref{lem:sparsify-before-contract}}
This section is devoted to prove \Cref{lem:sparsify-before-contract}, which we restate below.
\sparsify*

We apply the following standard sparsification result:
\begin{lemma}[Corollary 2.4 of \cite{Karger99sparsify}]\label{lem:sparsify}
Given an unweighted undirected graph $G$ with min-cut value $\lambda$ and any parameter $\delta\in(0,1)$, there exists $\alpha=O\left(\frac{\log n}{\delta^2\lambda}\right)$ such that if a subgraph $H$ is formed by picking each edge independently with probability $\alpha$ in $G$, then the following holds whp: for every cut $S$, its value in $H$ (denoted $d_H(S)$) and its value in $G$ (denoted $d_G(S)$) are related by $d_H(S)\in [(1-\delta)\cdot \alpha\cdot d_G(S), (1+\delta)\cdot \alpha\cdot d_G(S)]$. Note that this implies that the min-cut value in $H$ is $\lambda_H = O(\log n/\delta^2)$.

The running time of the sparsification algorithm is $O(m)$.
\end{lemma}

We apply the sparsification lemma (\Cref{lem:sparsify}) with parameter $\delta=\Theta\left(\frac{1}{\log n}\right)$ to obtain a sparsifier graph $H$. The graph $H$ is generated by picking each edge independently with probability $\alpha= \Theta\left(\frac{\log^3 n}{\lambda}\right)$ from $G$. 
We may assume wlog that $\lambda>\Omega(\log^3n)$, otherwise the statement is trivial with $H=G$ and $q=p$. So, $\alpha<1$.
$H$ has min-cut value $\lambda_H = O(\log /\delta^2) =O(\log^3 n)$, which is our desired property. 

We choose $q$ such that $1-q =\frac{1-p}{\alpha}$.
Note that $q < p$ because $\alpha < 1$ implies $1-q > 1-p$. 

It is clear that $u_H(q)$ is an unbiased estimator of $\ugp$, since keeping each edge with probability $1-p$ is equivalent to first choosing it with probability $\alpha$ and then keeping it with probability $1-q = \frac{1-p}{\alpha}$.

In the recursive case, we have  $n^{-1-O(1/\log \log n)} \le \ugp \le n^{-\Omega(1/\log \log n)}$ by \Cref{fact:recursive-case-ugp-range}.
Combined with \Cref{fact:ugp-basic-range}, we have $n^{-4}\le p^\lambda \le  n^{-\Omega(1/\log \log n)}$.
Denote $\tau = 1-p$. The above inequality implies
\begin{equation*}
\tau \le 1-e^{-4\ln n/\lambda}\le O\left(\frac{\log n}{\lambda}\right), \quad
1-q = \frac{\tau}{\alpha}\le  O\left(\frac{\frac{\log n}{\lambda}}{\frac{\log^3 n}{\lambda}}\right) = O\left(\frac{1}{\log^2 n}\right).
\end{equation*}

Let $Y_e$ be the indicator that edge $e$ is picked by the random graph ${H}$.  For any edge $e$, we have
\begin{align}
\E\left[q^{Y_e}\right] &= \alpha q + (1-\alpha) = 1-\alpha(1-q) = p \label{eq:qtoY} \\ 
\E\left[q^{2Y_e}\right] &= \alpha q^2 + (1-\alpha) = 1-\tau (1+ q) \le (1-\tau)(1-\tau q) = p\cdot (1-\tau q).\label{eq:qto2Y}
\end{align}
 We can bound $\E\left[q^{2Y_e}\right]$ in two ways:
\begin{align}
\frac{\E\left[q^{2Y_e}\right]}{p} &\le  1-\tau q \label{eq:qto2Yfirst}\\
\frac{\E\left[q^{2Y_e}\right]}{p^2} &\le  \frac{1-\tau q}{p} = \frac{(1-\tau)+\tau(1-q)}{1-\tau} = 1+\frac{\tau(1-q)}{1-\tau} \le 1+2\tau(1-q) = 1+O\left(\frac{1}{\lambda\log n}\right). \label{eq:qto2Ysecond}
\end{align}

Next we calculate the expectation and relative variance of $z_{H}(q)$. Notice that $Y_e$'s are independent for each edge $e$. Use $C_i\Delta C_j$ to denote the symmetric difference $(C_i\setminus C_j)\cup (C_j\setminus C_i)$ over two cuts $C_i, C_j$. Use $d_H(\cdot)$ to denote the cut value function in $H$.  First, we calculate the expectation of $z_{H}(q)$:
\[
\E[z_{H}(q)]
= \E\left[\sum_{C_i} q^{d_H(C_i)}\right] 
= \sum_{C_i}\E\left[q^{\sum_{e\in C_i}Y_e}\right]
= \sum_{C_i}\prod_{e\in C_i} \E\left[q^{Y_e}\right] 
\stackrel{(\ref{eq:qtoY})}{=} \sum_{C_i}p^{c_i}  
= z_G(p).
\]
Next, we bound the second moment of $z_{H}(q)$
\begin{align*}
    \E\left[(z_{H}(q))^2\right] 
    &= \E\left[\sum_{C_i}\sum_{C_j}q^{d_H(C_i)+d_H(C_j)}\right] 
    =\sum_{C_i}\sum_{C_j} \E\left[q^{\sum_{e\in C_i}Y_e + \sum_{e\in C_j}Y_e}\right] \\
    &=\sum_{C_i}\sum_{C_j}\prod_{e\in C_i\cap C_j}\E\left[q^{2Y_e}\right] \prod_{e\in C_i\Delta C_j} \E\left[q^{Y_e}\right]
    \stackrel{(\ref{eq:qtoY})}{=} \sum_{C_i}\sum_{C_j}p^{|C_i\Delta C_j|} \cdot \left(\E\left[q^{2Y_e}\right]\right)^{|C_i\cap C_j|}.
\end{align*}
We partition this sum into three parts and separately bound their ratios with $(\zgp)^2$.

For terms with $C_i=C_j$,
\begin{align}
&\frac{\sum_{C_i} p^{|C_i\Delta C_i|} \cdot \left(\E\left[q^{2Y_e}\right]\right)^{|C_i\cap C_i|}}{(\zgp)^2} 
= \frac{\sum_{C_i} \left(\E\left[q^{2Y_e}\right]\right)^{c_i}}{(\zgp)^2} 
= \frac{\sum_{C_i} p^{c_i}\cdot \left(\frac{\E\left[q^{2Y_e}\right]}{p}\right)^{c_i}}{\sum_{C_i} p^{c_i} \cdot \zgp} 
\le \max_{C_i} \frac{\left(\frac{\E\left[q^{2Y_e}\right]}{p}\right)^{c_i}}{\zgp} \nonumber \\
&\stackrel{(\ref{eq:qto2Yfirst})}{\le} \frac{(1-\tau q)^\lambda}{\zgp} 
\le \left(\frac{1-\tau q}{p}\right)^\lambda
\stackrel{(\ref{eq:qto2Ysecond})}{\le} \left(1+O\left(\frac{1}{\lambda\log n}\right)\right)^\lambda = 1+O\left(\frac{1}{\log n}\right). \label{eq:equal-terms-summation}
\end{align}

For terms with $|C_i\cap C_j|\le \lambda$,
\begin{align*}
&\frac{\sum_{|C_i\cap C_j|\le\lambda} p^{|C_i\Delta C_j|}\cdot \left(\E\left[q^{2Y_e}\right]\right)^{|C_i\cap C_j|}}{(\zgp)^2}
=\frac{\sum_{|C_i\cap C_j|\le \lambda}\, p^{c_i+c_j}\cdot \left(\frac{\E\left[q^{2Y_e}\right]}{p^2}\right)^{|C_i\cap C_j|}}{(\zgp)^2} \\
&\stackrel{(\ref{eq:qto2Ysecond})}{\le} \frac{\sum_{|C_i\cap C_j|\le \lambda}\, p^{c_i+c_j}\left(1+O\left(\frac1{\lambda\log n}\right)\right)^{\lambda}}{\sum_{C_i,C_j}p^{c_i+c_j}}  
\le \left(1+O\left(\frac{1}{\lambda\log n}\right)\right)^\lambda = 1+O\left(\frac{1}{\log n}\right).
\end{align*}

For terms with $|C_i\cap C_j| > \lambda$ and $C_i\ne C_j$, we have 
\begin{align*}
&\frac{\sum_{C_i\ne C_j,|C_i\cap C_j|>\lambda} \, p^{|C_i\Delta C_j|}\cdot \E\left[q^{2Y_e}\right]^{|C_i\cap C_j|}}{(\zgp)^2} 
= \frac{\sum_{C_i\ne C_j,|C_i\cap C_j|>\lambda} \, p^{|C_i\cup C_j|} \cdot \left(\frac{\E\left[q^{2Y_e}\right]}{p}\right)^{|C_i\cap C_j|}}{(\zgp)^2} \\
&\stackrel{(\ref{eq:qto2Yfirst})}{\le} \frac{\sum_{C_i\ne C_j,|C_i\cap C_j|>\lambda} \, p^{|C_i\cup C_j|}\cdot (1-\tau q)^{\lambda}}{(\zgp)^2} 
\le \frac{x_G(p) \cdot (1-\tau q)^{\lambda}}{(\zgp)^2} \quad \text{(by definition of $\xgp$)}.
\end{align*}
Applying $\frac{\xgp}{\zgp} \le \frac1{\log n}$ from \Cref{lem:z-approx-u}, this is at most
\[\frac{1}{\log n}\cdot \frac{(1-\tau q)^\lambda}{\zgp} 
\stackrel{(\ref{eq:equal-terms-summation})}{\le} \frac{1}{\log n}\cdot \left(1+O\left(\frac{1}{\log n}\right)\right) 
\le O\left(\frac1{\log n}\right) .\]
In conclusion, the total relative second moment of $z_H(q)$ is given by
\[\frac{\E[z_{H}(q)^2]}{(\zgp)^2} \le \left( 1+O\left(\frac1{\log n}\right)\right)+\left( 1+O\left(\frac1{\log n}\right)\right)+O\left(\frac1{\log n}\right)=2 + O\left(\frac1{\log n}\right).
\]
Finally, we can apply \Cref{lem:u-from-z} to bound the relative second moment of $u_H(q)$ by $2 + O\left(\frac1{\log n}\right)$.


\end{document}